%% file: main.tex
\documentclass{article}
\usepackage{styleFile}

\title{Quantifying model prediction sensitivity to model-form uncertainty}
\author[1]{Teresa Portone\thanks{Corresponding author: tporton@sandia.gov}}
\author[1]{Rebekah D. White}
\author[1]{Joseph L. Hart}
\affil[1]{Sandia National Laboratories}

\date{}

\begin{document}

\maketitle
\input{abstract.tex}
\input{introduction.tex}
\input{BackMeth.tex}
\input{results.tex}

\input{conclusions}
\input{acknowledgements.tex}

\appendix
\input{Appendix_A.tex} 
\input{Appendix_B.tex}

\input{Appendix_C.tex}

\singlespacing
\bibliographystyle{unsrturl} 
\bibliography{references}

\end{document}

%% file: abstract.tex
\begin{abstract}
Model-form uncertainty (MFU) in assumptions made during physics-based model development is widely considered a significant source of uncertainty;
however, there are limited approaches that can quantify MFU in predictions extrapolating beyond available data.
As a result, it is challenging to know how important MFU is in practice, especially relative to other sources of uncertainty in a model, making it difficult to prioritize resources and efforts to drive down error in model predictions. 
To address these challenges, we present a novel method to quantify the importance of uncertainties associated with model assumptions. 
We combine parameterized modifications to assumptions (called MFU representations) with grouped variance-based sensitivity analysis to measure the importance of assumptions.
We demonstrate how, in contrast to existing methods addressing MFU, our approach can be applied without access to calibration data.
However, if calibration data is available, we demonstrate how it can be used to inform the MFU representation, and how variance-based sensitivity analysis can be meaningfully applied even in the presence of dependence between parameters (a common byproduct of calibration).
\end{abstract}

%% file: introduction.tex
\section{Introduction}\label{sec:introduction}
\firstline{In high-consequence science and engineering problems it is often necessary to rely on physics-based mathematical models to predict the behavior of a complex system or phenomenon of interest, e.g.,~making predictions about the future or for conditions where experiments are impossible.}
However, in the process of model development, assumptions must be made that can degrade the trustworthiness of these models. 
Such assumptions include, but are not limited to, assumptions regarding invariance under transformation (e.g. rotation: isotropy, translation: homogeneity), couplings between phenomena, and dependencies (e.g., ignoring microscale dependence, assuming rate dependence).

\firstline{Uncertainty associated with the mathematical form and validity of these assumptions leads to \textit{model-form uncertainty (MFU)}, where the exact mathematical representation needed to model a phenomenon of interest is unknown or cannot be represented using the information present in the model (e.g., dependence on quantities below the length-scale resolved in the model).}
MFU and other sources of uncertainty, such as parameter uncertainty, result in uncertainties in model predictions that must be quantified to adequately characterize the trustworthiness of resulting predictions.
Being able to identify the most important sources of uncertainty is critical for determining where to invest resources to most improve prediction reliability through model improvement, uncertainty reduction, etc.
To this end, this work presents a first-of-its-kind method to quantify the impact of modeling assumptions on model predictions relative to other sources of uncertainty.

\firstline{There are a variety of ways to represent uncertainty in a model's representation of physics.}
The first established and most common approach was presented in~\cite{kennedy_bayesian_2001}, wherein a statistical model is appended to the output of a physical model, and its hyperparameters are calibrated to minimize misfit with data. 
In~\cite{asme_committee_vv_20_standard_2009} MFU is discussed and a process to evaluate it for a given data set is provided, but no guidance is provided on how such results could be interpolated (or extrapolated) beyond that data set.
In \cite{sargsyan_statistical_2015}, MFU is characterized by augmenting uncertainty in a modeling assumption’s parameterization to account for discrepancy between data and predictions without perturbing the model assumption’s mathematical form.
In \cite{acquesta_model-form_2022}, a data-driven modification to modeling assumptions is used to represent MFU. 
In \cite{subramanian_bayesian_2019,subramanian_error_2019,subramanian_model_2020} a statistical model for the time evolution of a correction to the assumption is inferred using time-series data.
All the aforementioned approaches rely on access to observational data, which is often unavailable in scenarios where model predictions are most needed, e.g.,~to inform policy and engineering design decisions.

\firstline{This work instead leverages what we term herein \textit{MFU representations}, which characterize MFU through parameterized modifications to modeling assumptions~\cite{oliver_validating_2015,morrison_representing_2018,portone_representing_2019,bandy_quantifying_2024,bandy_stochastic_2024,bandy_stochastic_2025}.}
Through their parameterizations, MFU representations are constrained to respect known phenomenology and relevant physics, such as conserved quantities, with the aim of maximizing their extrapolative capability.
To represent uncertainty, their parameterizations are assigned probability distributions, where relevant physical constraints and phenomenological behavior may also be encoded.
The resulting MFU representation provides a range of plausible behavior for the modeled phenomenon. 
Because they are developed to respect known physics and phenomenology, MFU representations can be developed without access to calibration data, which is a common situation in the preliminary stages of model development and uncertainty analysis.
However, if data is available, MFU representations can be informed using, e.g., Bayesian inference.
Hierarchical frameworks that calibrate the hyperparameters of an MFU parameterization are common, so we demonstrate one such a framework here.

\firstline{We use grouped variance-based sensitivity analysis (VBSA) to quantify the impact of physics modeling assumptions on model predictions herein.}
Namely, grouped Sobol' indices measure the fraction of variance in a model prediction that can be attributed to each input factor in a model. 
Standard practice has focused on individual model parameters as input factors, but an input factor can also be defined as a group of parameters~\cite{sobol_1993, sobol_2001, prieur_2017}.
In the context of a group of parameters, Sobol' indices measure the fraction of variance in a model prediction that can be attributed to all parameters in the group as well as to their interactions with each other. 

\firstline{In this work, we leverage grouped Sobol' indices to  provide a single sensitivity measure for the group of parameters comprising an MFU representation, relative to other sources of uncertainty in the model.}
Since multiple MFU representations may be valid in a given modeling scenario, we establish bounds for the differences in Sobol' indices corresponding to varying MFU representations.
Furthermore, we demonstrate how grouped VBSA can be used to analyze sensitivity to MFU when hierarchical probability densities are used, even after calibration, which often induces dependence between MFU and other model parameters. 

\firstline{While the proposed approach provides a measure of importance of modeling assumptions to model predictions, it can also measure the importance of model components/subcomponents whose parameterizations can be grouped in an uncertainty analysis.} 
Thus, our proposed methodology can be used to inform where to invest resources, e.g., more experimentation or computation, to drive down the most impactful sources of uncertainty on model predictions.
Furthermore, following the spirit of~\cite{li_role_2016,paquette-rufiange_optimal_2023}, it can be used to assess whether assumptions important to prediction use cases are also important in validation cases as a means of assessing the relevance of validation evidence.
The key contributions of this work include:
\begin{itemize}[label=\raisebox{0.2ex}{\tiny$\bullet$}, leftmargin=*, itemsep=0.5ex, parsep=0pt]
\item Presenting a novel approach to measure the impact of model-form uncertainty on model predictions using variance-based sensitivity analysis.
\item Proving that grouped variance-based sensitivity analysis is robust to an MFU representation's parameterization.
\item Analyzing computation of Sobol' indices for hierarchical distributions.
\end{itemize}

The remainder of the paper proceeds as follows.
\Cref{sec:back_methods} presents the methodology of our proposed approach and provides the relevant background related to MFU representations (\Cref{sec:MFU_representations}) and their calibration (\Cref{sec:MFU_calibration})  as well as
grouped Sobol' indices (\Cref{sec:grouped_indices}); an illustrative example is provided to build intuition.
In~\Cref{sec:robustness}, we establish robustness of the approach to varying MFU representations.
We demonstrate our approach on a practical application problem of upscaled contaminant transport in~\Cref{sec:results}.
Final conclusions are presented in~\Cref{sec:conclusions}.

%% file: BackMeth.tex
\section{Background and methods}
\label{sec:back_methods}

\firstline{In this section, we provide the background and methodology needed to measure the impact of model-form uncertainty (MFU) on model predictions using variance-based sensitivity analysis (VBSA).}
We first provide the mathematical framework for MFU representations along with a discussion of how this approach differs from alternative model discrepancy approaches~\cite{kennedy_bayesian_2001} in~\Cref{sec:MFU_representations}. 
In~\Cref{sec:MFU_calibration} we discuss how data can be used to inform such representations through a hierarchical framework for calibration.
Finally, \Cref{sec:grouped_indices} provides the notation and formulation of the grouped Sobol' indices used to compute the sensitivity of the modeling assumption.

\input{mfu_back.tex}
\input{mfu_cal.tex}
\input{sobol_back.tex}
\input{sobol_est.tex}

%% file: mfu_back.tex
\subsection{Model-form uncertainty representations}\label{sec:MFU_representations}
\firstline{Model-form uncertainty (MFU) representations provided parameterized, embedded modifications to mathematical models that allow one to characterize uncertainty in modeling assumptions.} 
First introduced in~\cite{oliver_validating_2015} and applied in, e.g.,~\cite{morrison_representing_2018,portone_representing_2019,bandy_quantifying_2024,bandy_stochastic_2024,bandy_stochastic_2025}, MFU representations augment the governing equations~\cite{morrison_representing_2018,bandy_quantifying_2024,bandy_stochastic_2024,bandy_stochastic_2025} or replace an invalid assumption entirely~\cite{portone_representing_2019}.

Defining concepts mathematically, let a model governing the evolution of state variables $v$ and model output mappings be denoted
\begin{align}
        0 &= \mathcal{R}_a(v), \label{eq:governing_eq}\\ 
        d &= \mathcal{O}(v) + \epsilon_m, \quad \epsilon_m \sim \pi(\epsilon_m),\label{eq:obs_op} \\
        q &= \mathcal{Q}(v), \label{eq:qoi_op}
\end{align}
where~\eqref{eq:governing_eq} represents a set of governing equations in their residual form;
~\eqref{eq:obs_op} is the mapping of $v$ to observable data $d$ through an observation operator $\mathcal{O}$, with measurement error $\epsilon_m$ typically assumed to be a random variable with distribution $\pi(\epsilon_m)$;
and~\eqref{eq:qoi_op} is the mapping of $v$ to a predicted quantity of interest (QoI), which may not be observable (e.g., a prediction of the future). 
In~\eqref{eq:governing_eq}, $a$ is an aspect of the governing equations for which an assumption must be introduced, e.g.,~a coupling mechanism or the mathematical form for a constitutive equation. 
The assumption may not explicitly appear in the governing equations (e.g.,~neglecting a term), however, we explicitly denote its existence here to signify that a MFU representation will be introduced to characterize its impact on model predictions.

\firstline{In physics-based models, the governing equations $\mathcal{R}$ are generally highly trusted theory such as conservation of mass/momentum/energy.}
However, if the assumption is inaccurate, the governing equations will not provide a correct representation of reality; that is, the inaccurate assumption creates \textit{model-form error}.
The inaccuracy of the assumption will propagate through the governing equations to $v$ and will be observed as a discrepancy between the predicted observable output $\mathcal{O}(v)$ and $d$ that cannot be explained merely by the presence of measurement error, a phenomenon known as \textit{model discrepancy}.
Mathematically, this amounts to the difference between observed data and model output not obeying the assumed probability distribution of the measurement error:
\begin{equation*}
    d-\mathcal{O}(v) \not\sim \pi(\epsilon_m).
\end{equation*}
Calibration of model parameters without accounting for model-form error can lead to biased, overconfident uncertainties, resulting in flawed representations of predictive uncertainties.
Example~\ref{ex:poly} illustrates this phenomenon for a simple problem.
\begin{examplebox}
\label{ex:poly}
Consider we have a ``true'' data-generating process given by
\begin{align}
  \label{eq:f_data}
  f_{true}(x) = c_0 + c_1x + 0.1(x^2 + x^3), \quad x \in [0,2], \quad c_0=c_1=1.
\end{align}
Measurement data is collected at $N_d=100$ equally-spaced points in $[0,2]$ with measurement uncertainty $\epsilon_i$:
\begin{align}
  d_i = f_{true}(x_i) + \epsilon_i, \quad \epsilon_i \sim \pi_\epsilon(\epsilon_i) = \mathcal{N}(0, (0.05)^2), \quad i=1, \ldots, N_d.
\end{align}

However, without knowledge of this true model, one may assume a linear dependence on $x$:
\begin{align}\label{eq:assum_f}
  f(x;c_0, c_1) = c_0 + c_1 x,
\end{align}
with uncertain parameters $c_0$ and $c_1$.
This assumed linearity is the type of assumption denoted by the $a$ in~\eqref{eq:governing_eq}.

\vspace{1em}
We wish to use Bayesian inference to inform these uncertain parameters from data $\mathbf{d} = d_1, \dots, d_{N_d}$.
Given the prior belief that $c_0 \sim \Ucal[0,2]$ and $c_1 \sim \Ucal[0,5]$ and 
likelihood density,
\begin{align}
    p(\mathbf{d} | c_0, c_1 ) = \pi_\epsilon\left( \mathbf{d} - f(\mathbf{x};c_0, c_1) \right),
\end{align}
which provides the probability that the data arose from the model for given parameter values $c_0$ and $c_1$, the Bayesian posterior density is:
\begin{align}
    p(c_0, c_1 | \mathbf{d} ) = \frac{p(\mathbf{d}|c_0,c_1)p(c_0)p(c_1)}{p(\mathbf{d})}.
\end{align}
The posterior density can be approximated numerically using algorithms such as Markov Chain Monte Carlo (MCMC). 

\vspace{1em}
Bayesian calibration of $c_0$ and $c_1$ yields the marginal posteriors depicted in~\Cref{fig:inad_post_vs_pri}.
Here, we see that the model-form error in $f(x)$ from the neglected nonlinearity causes the posterior to concentrate around incorrect values for $c_0$ and $c_1$, which for many problems are the maximum likelihood estimate (MLE) values that minimize misfit with the data~\cite{bjk_kleijn_bernstein-von-mises_2012}.
\begin{center}
  \includegraphics{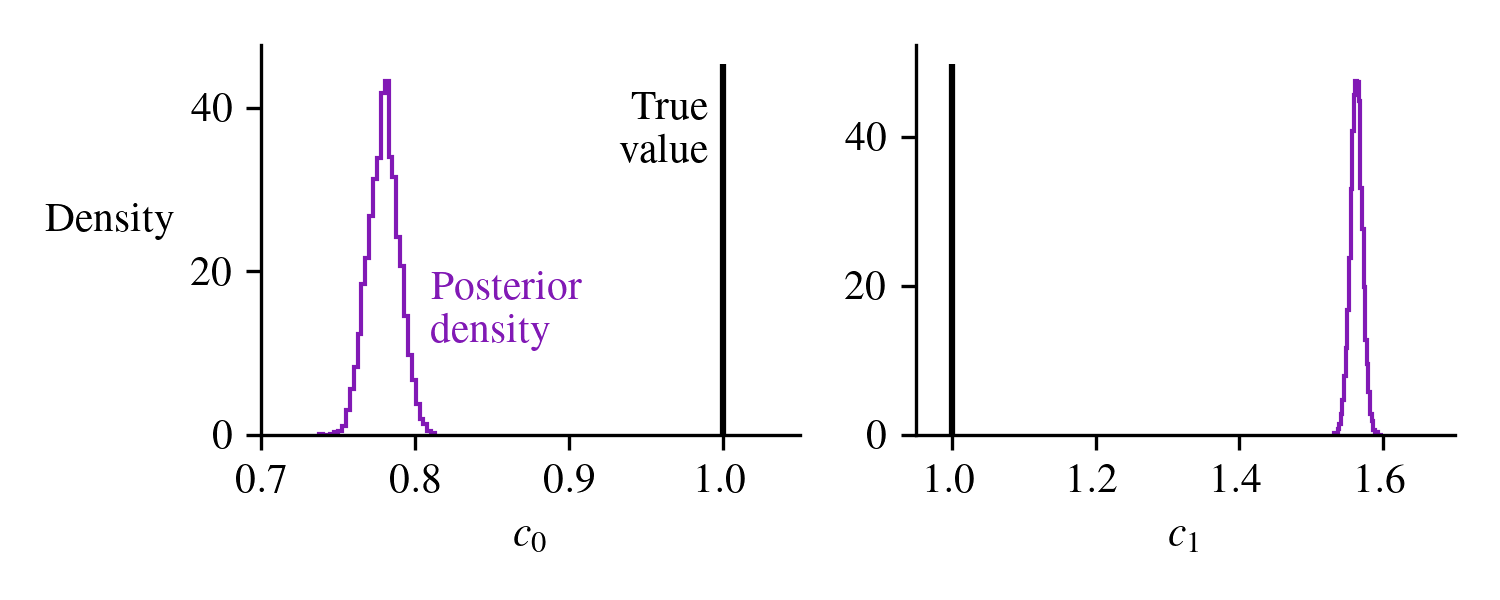}\vspace{-1em}
  \captionof{figure}{Posterior density histograms for $c_0$ and $c_1$, compared to their true values.}
  \label{fig:inad_post_vs_pri}
\end{center}
\Cref{fig:inad_predictive} shows prior and posterior predictive distributions, which represent the distribution of future unseen data accounting for model (e.g., parameter) and measurement uncertainties. 
A sample from a predictive distribution is generated by generating a sample of input uncertainties, evaluating the model for that sample, and perturbing the evaluation with a sample of the measurement error.
\Cref{fig:inad_predictive} illustrates how utilizing an inadequate model in the Bayesian inverse problem results in a posterior predictive distribution whose 95\% confidence intervals fail to encompass the data. 
Despite the fact that extrapolation to larger $x$ values further increases the discrepancy between the true phenomenon and the posterior predictive distribution, the posterior uncertainty bounds would not indicate any increased uncertainty as a result of this model-form error.
\begin{center}
  \includegraphics{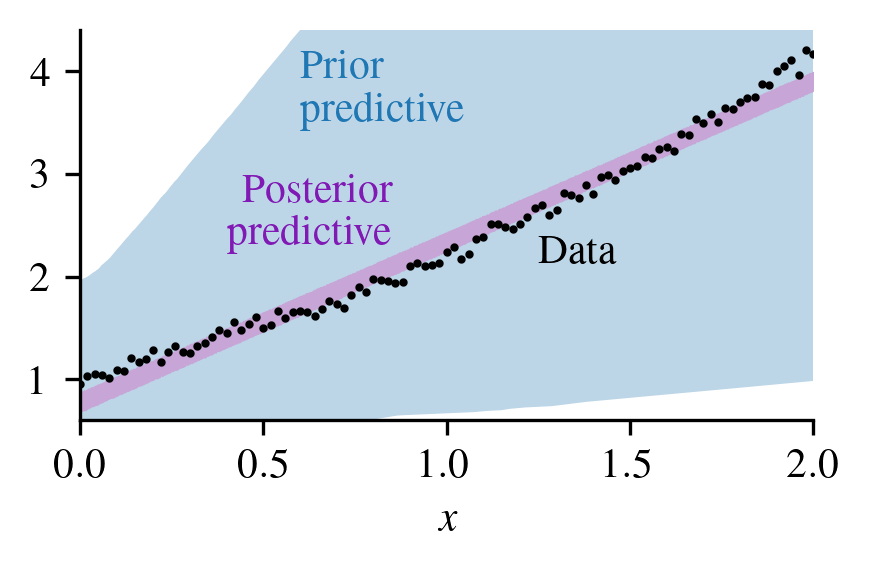}
  \includegraphics{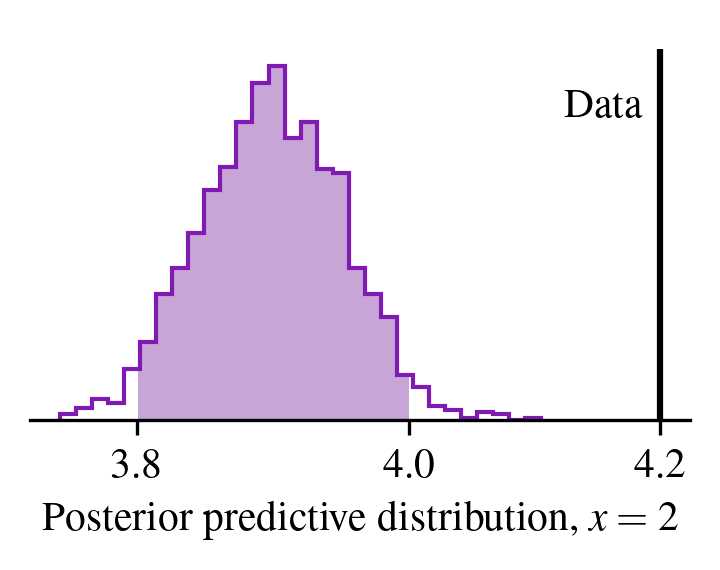}\vspace{-1em}
  \captionof{figure}{Left: 95\% probability mass intervals based on 0.025 and 0.975 quantiles for the prior and posterior predictive distributions corresponding to the mathematical model given in~\eqref{eq:assum_f} compared to calibration data. Right: Data vs.~posterior predictive distribution at $x=2$ with 95\% probability mass interval shaded in.}
  \label{fig:inad_predictive}
\end{center}
\vspace{1em}
\end{examplebox}

\firstline{Approaches addressing model discrepancy, based off the seminal paper~\cite{kennedy_bayesian_2001}, append a statistical model (usually a Gaussian process) $\delta_\phi$ to the observable model output and calibrate its hyperparameters $\phi$ to minimize misfit with the data:}
\begin{align}
    0 &= \mathcal{R}_a(v),\\ 
    d &= \mathcal{O}(v) + \epsilon_m + \delta_\phi, \quad \epsilon_m \sim \pi(\epsilon_m),\label{eq:ko} \\
    q &= \mathcal{Q}(v).
\end{align} 
While this approach is attractive as it is completely nonintrusive to the model, it is primarily useful in cases where the prediction quantity of interest is an observable output and predictions fall within the range of available calibration data, so that the statistical model is interpolating and not extrapolating~\cite{maupin_model_2020}.
Even interpolation is not guaranteed to be accurate, however~\cite{brynjarsdottir_learning_2014}. 
Since the correction to the model output only appears in~\eqref{eq:ko}, there is no way to propagate the discrepancy representation to an unobservable QoI (e.g., quantities representing the interior of modeled systems where sensors cannot be placed). 
Furthermore, even if prediction QoIs are observable outputs and a discrepancy representation can be calibrated, because it is purely statistical, its accuracy will be limited outside the regime where data is available.

\firstline{The work proposed herein is motivated by the fact that prediction QoIs are often unobservable, meaning predictions must extrapolate beyond available data.}
In such cases, only embedded MFU representations, which address the model-form error at its source, can propagate to the QoIs. 
To this end, an MFU representation $\xi_\gamma$ replaces or augments $a$ to represent uncertainty in the assumption, thereby modifying the governing equations:
\begin{align}
    0 &= \mathcal{R}_{\xi_\gamma}(v), \label{eq:mfu_gov}\\ 
    d &= \mathcal{O}(v(\xi_\gamma)) + \epsilon_m, \quad \epsilon_m \sim \pi(\epsilon_m), \label{eq:mfu_obs}\\
    q &= \mathcal{Q}(v(\xi_\gamma)). \label{eq:mfu_qoi}
\end{align}
Here $\gamma$ denotes the parameters of the MFU representation, and we denote the state variables $v(\xi_\gamma)$ in \eqref{eq:mfu_obs} and~\eqref{eq:mfu_qoi} to emphasize the fact that $\xi_\gamma$'s inclusion in~\eqref{eq:mfu_gov} propagates through $v$ to the observable outputs $\mathcal{O}(v)$ and prediction QoIs $\mathcal{Q}(v)$. 
Hence, any information gained by calibration against observable data will be propagated to unobservable prediction QoIs.

\firstline{MFU representations are developed with prediction in mind.}
Since the MFU representation $\xi_\gamma$ is embedded in the governing equations and directly characterizes uncertainty in a modeling assumption, it can be constrained to respect physical properties of the modeled system and reflect what is known about the modeled phenomenon.
Because more intuition can be brought to bear by directly characterizing uncertainty in the assumption, MFU representations can be developed without access to data, which can be especially useful early in the model development process.
Furthermore, uncertainty in modeling assumptions can be encoded through the probability distributions assigned to the MFU parameters, which are defined to reflect what is known about plausible behavior for the modeled phenomenon.

\firstline{Careful treatment of uncertainties is needed when data are available to inform the MFU representation.}
MFU representations must be able to represent uncertainty arising from the fact that, often, no single mathematical model can perfectly capture a complex phenomenon that in reality depends upon quantities that aren't resolved in the model (for example, neglected species in a chemistry model~\cite{morrison_representing_2018} or a population model~\cite{bandy_stochastic_2024}). 
No amount or quality of data can overcome such a structural inadequacy in the model; the goal must then be to use data to inform the appropriate degree of uncertainty in the model's form.

%% file: mfu_cal.tex
\subsection{Calibration of model-form uncertainty representations}\label{sec:MFU_calibration}

\firstline{Although data is not required to develop an MFU representation, in practice calibration data is often used to inform uncertainty in the MFU parameters.}
Standard Bayesian paradigms for estimating uncertainties are often unsuitable for calibrating MFU parameterizations;
in the limit of infinite data, the Bayesian posterior contracts about parameter values that best agree with the data~\cite{bjk_kleijn_bernstein-von-mises_2012}, implying that uncertainty vanishes. 
In practice, such posterior contraction could allow one to grow increasingly confident in an incorrect mathematical model form.
Thus, when informing MFU representations with data, one should employ an inference paradigm capable of avoiding posterior contraction when model-form error is present.

\firstline{
Hyperparameterization is a common approach to prevent underpredicting uncertainty in MFU characterizations.
}
One approach to hyperparameterization, presented in~\cite{sargsyan_embedded_2019}, augments model parameters $\theta$ with ``discrepancy'' random variables, which are stochastic: $\theta + \delta, \; \delta \sim \pi_\delta(\cdot;\phi)$.
The joint distribution of parameters and hyperparameters $\pi(\theta,\phi)$ is then calibrated. 
Another common approach, as presented in~\cite{oliver_validating_2015,sargsyan_statistical_2015,bandy_stochastic_2024}, leverages hierarchical representations,
where the MFU parameters are random variables described by parametric probability density functions (PDFs) whose hyperparameters are also uncertain (i.e., have corresponding PDFs).
For example, if an MFU parameter $\gamma$ is assumed to be normally distributed, its mean and standard deviation are considered uncertain with assumed PDFs as well.
Denoting such hyperparameters as $\phi$, the joint distribution of parameters and hyperparameters are defined hierarchically as 
\begin{align}
    \pi(\gamma,\phi) = \pi(\gamma|\phi)\pi(\phi).\label{eq:hierarchical_distribution}
\end{align}

\firstline{When considering hierarchical approaches to inform MFU representations, calibration using a fully Bayesian treatment of uncertainty leverages observational data $d$ to estimate a posterior distribution for $\phi$ alone~\cite{oliver_validating_2015,morrison_representing_2018,bandy_quantifying_2024,bandy_stochastic_2025,bandy_stochastic_2024}:}
\begin{align}
    \pi(\phi|d) = \frac{\pi(d|\phi) \pi_0(\phi)}{\pi(d)}.
    \label{eq:hyperparameter_posterior}
\end{align}
An analytical expression for the likelihood with respect to hyperparameters $\pi(d|\phi)$ is not typically available, so practical implementations perform a joint Bayesian calibration of parameters and hyperparameters:
\begin{align*}
    \pi(\gamma,\phi|d) = \frac{\pi(d|\gamma,\phi) \pi_0(\gamma,\phi)}{\pi(d)} 
    =\frac{\pi(d|\gamma) \pi(\gamma|\phi)\pi(\phi)}{\pi(d)},
\end{align*}
where the model parameters are marginalized out in postprocessing.

Because the hyperparameters don't appear directly in the model, forward propagation is carried out by sampling the hierarchical distribution~\eqref{eq:hierarchical_distribution} to generate samples of the MFU parameters.
The prior pushforward of hyperparameters to MFU parameters uses the hyperprior, $\pi_0(\phi)$.
The posterior pushforward of hyperparameters to MFU parameters uses the posterior distribution in hyperparameters, $\pi(\phi|d)$.
An example of the hierarchical Bayesian approach to inform MFU is presented in Example~\ref{ex:hier}.

\begin{examplebox}
\label{ex:hier}

The discrepancy between calibration data and the posterior predictive uncertainty bounds in~\Cref{fig:inad_predictive} indicate model-form error is present. 
We may suspect that the data generating process is weakly nonlinear, but we don't know the exact form of the nonlinearity.
In this case an MFU representation of the following form could be proposed:
\begin{align}
  \tilde{a}_m(c_2, \alpha) = c_2 x^\alpha,
  \label{eq:poly_mfu}
\end{align}
yielding an enriched model
\begin{align}
  \label{eq:f_enriched}
  \ftilde(x) = c_0 + c_1x + c_2 x^\alpha
\end{align}
with uncertain model parameters $[c_0, c_1]$ and MFU parameters$[c_2, \alpha]$. 
We expect~\eqref{eq:poly_mfu} does not perfectly capture the true nature of the nonlinearity, and thus, the standard Bayesian posterior predictive may not encompass the calibration data.
We therefore leverage a hierarchical approach, wherein we hyperparameterize the problem in terms of $\muc, \sigc, \mua$, and $\siga$ and instead pose the Bayesian inverse problem with respect to model parameters $[c_0,c_1]$ and MFU hyperparameters $[\muc, \sigc, \mua, \siga]$.
Since, $c_2 > 0$ and $\alpha > 1$, we pose prior distributions that enforce these constraints:
\begin{align*}
  \log(c_2) & \sim \Ncal(\muc, \sigc^2), \\
  \log(\alpha-1) & \sim \Ncal(\mua, \siga^2).  
\end{align*}
Since we suspect a weak nonlinearity, we do not expect large values for $c_2$ and $\alpha$. 
Therefore, we define hyperpriors such that the supports for marginal prior distributions of $c_2$ and $\alpha$ fall mostly within $[0,1]$ and $[1,4]$, respectively:
\begin{eqnarray*}
  \begin{aligned}
  \muc & \sim \Ncal(-1, 0.5^2),\\ 
  \sigc & \sim \Ucal[0,0.1],
  \end{aligned}
  \hspace{0.5in}
  \begin{aligned}
    \mua & \sim \Ncal(0, 0.5^2),\\ 
    \siga & \sim \Ucal[0,0.1].
    \end{aligned}
\end{eqnarray*}

\Cref{fig:poly_with_mfu_prior_vs_post} depicts the marginal posterior distributions for $c_0$, $c_1$; by incorporating the MFU representation and leveraging hierarchical inference, the posterior marginals for these model parameters now encompass their true values.
\begin{center}
  \includegraphics{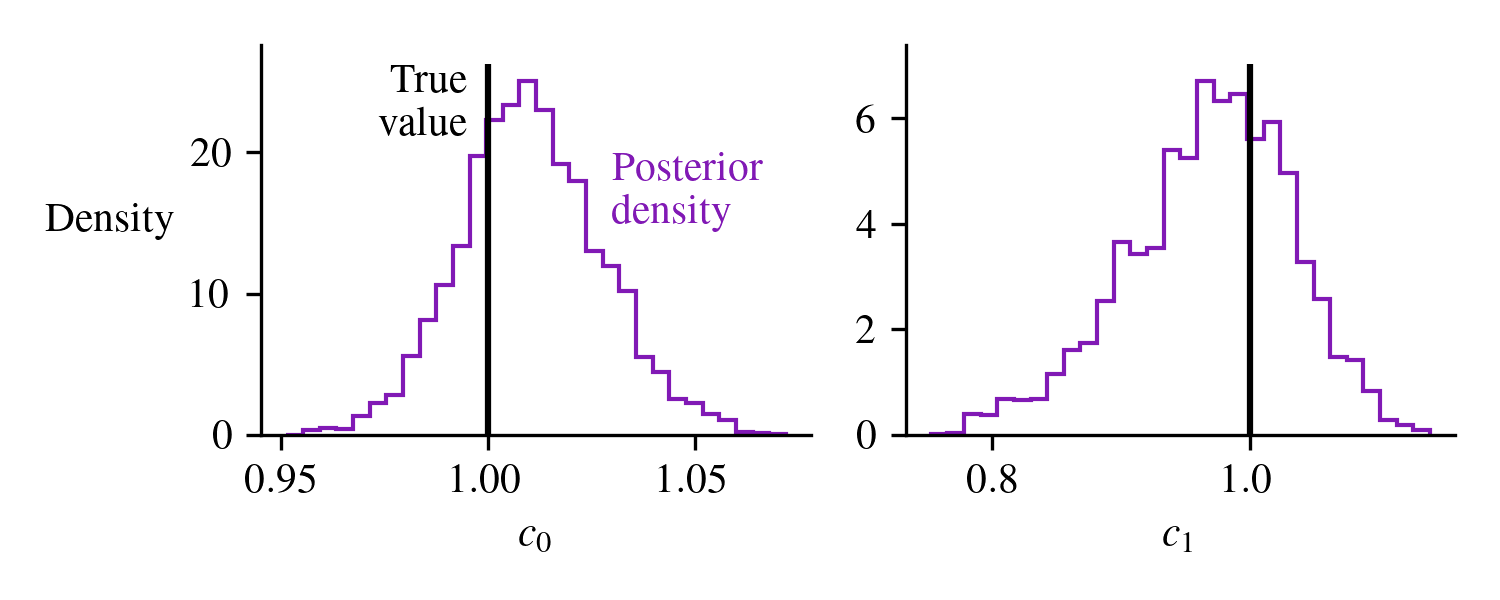}\vspace{-1em}
  \captionof{figure}{Posterior density histograms for $c_0$ and $c_1$ compared to their true values.}
  \label{fig:poly_with_mfu_prior_vs_post}
\end{center}
Furthermore, we see in~\Cref{fig:poly_with_mfu_predictives} that by adequately quantifying the MFU arising from the weak nonlinearity, the resulting 95\% confidence intervals for the posterior predictive now encompass the data.
\begin{center}
  \includegraphics{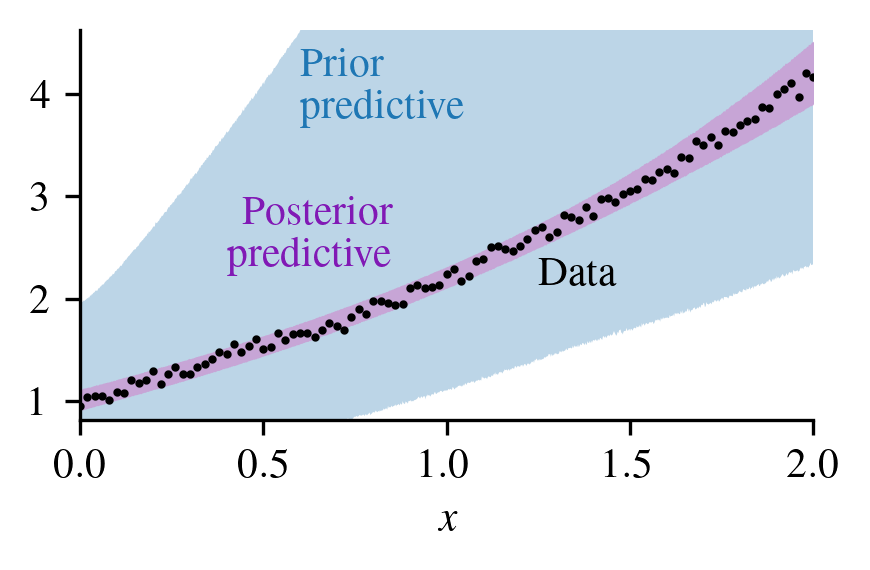}
    \includegraphics{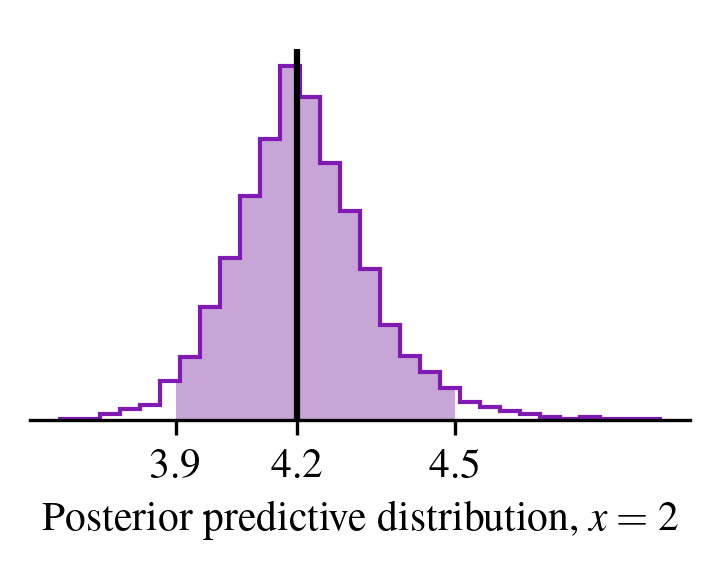}
  \vspace{-1em}
  \captionof{figure}{Left: 95\% probability mass intervals based on 0.025 and 0.975 quantiles for the prior and posterior predictive distributions corresponding to the enriched mathematical model given in~\eqref{eq:f_enriched}  compared to calibration data. Right: Data vs.~posterior predictive distribution at $x=2$ with 95\% probability mass interval shaded in.}
  \label{fig:poly_with_mfu_predictives}
\end{center}
As compared to the posterior and posterior-predictive distributions attained when calibrating with an inadequate model,~\Cref{fig:poly_with_mfu_prior_vs_post,fig:poly_with_mfu_predictives} are less biased and better capture the calibration data within the confidence bounds of the posterior predictive distribution.
This illustrative example demonstrates that inclusion of a hierarchically calibrated MFU representation can improve posterior accuracy and avoid overconfidence in posterior predictive distributions.
\end{examplebox}

%% file: sobol_back.tex
\subsection{Grouped Sobol' indices}
\label{sec:grouped_indices}

\firstline{Grouped Sobol' indices provide a measure of importance of groups (i.e., subsets) of model parameters to prediction quantities of interest (QoIs).} 
Specifically, Sobol' indices are a global sensitivity analysis method that quantify the fraction of variance in a model output that can be attributed to an uncertain model input (or group of inputs).
Sobol' indices were originally defined in the context of the functional analysis of variance (ANOVA) decomposition~\cite{sobol_1993}, which
we do not review here for the sake of brevity.
Extensive discussion regarding Sobol' indices and their computation can be found in, e.g.,~\cite{sobol_1993, sobol_2001, saltelli_2010, prieur_2017}.
Here we define grouped Sobol' indices and methods for their computation.

First we define the mathematical notation for the remainder of the paper.
Let $\inpRV_i$ be a random variable with support $\Omega_i \subset \mathbb{R}$, for $i = 1, 2, \dots, \inpdim$. 
Consider $f \in  \mathcal{L}^2$ where $f: \Omega \to \mathbb{R}$ for $\Omega = \Omega_1 \times \Omega_2 \times \dots \times \Omega_{\inpdim}$.
We denote the indices for parameter group $\ug = \{i_1, i_2, \dots, i_k\} \subset \{1, 2, \dots, \inpdim\}$ and the group of inputs $\binpRV_{\ug} = \left( \inpRV_{i_1}, \inpRV_{i_2}, \dots, \inpRV_{i_k}\right)$.
We denote the complement of the group $\binpRV_{\sim \ug} = \binpRV_{ \{1,2,\dots,\inpdim \} \setminus \{i_1, i_2, \dots, i_k\}}$.
The joint probability density for $\X_{\bm u}$ is denoted $\pi(\x_{\bm u})$ (for notational simplicity, we omit the distribution subscripts in $\pi_{\bm u}(\x_{\bm u})$).

The parameters within a given group may be statistically dependent on each other. 
However, groups are independent of one another and other inputs. 
Hence, the joint probability density for $\X$ can be written as the product of each group, e.g., $\pi(\x) = \pi(\x_{\bm u})\pi(\x_{\sim \ug})$. 
As is common, we assume that the parameters are statistically independent to facilitate the derivation of the Sobol' indices. 
However, we will relax this assumption in~\Cref{sec:dependent_groups_SA} to accommodate dependencies that arise as a result of calibration.

The \textit{grouped first-order Sobol' index} (or \textit{grouped main effect index}) for parameter group $\ub$ is defined as
\begin{align}
    S_{\ub}^{g} = \frac{\Var_{\X_\ub}\left(\E_{\X_{\sim \ub}}\left[f(\X)|\X_\ub\right] \right)}{\Var(f(\X))}.
    \label{eq:grouped_main_effect}
\end{align}
This grouped index is equivalent to the \textit{closed Sobol' index} (often denoted $S_\ub^{\text{clo}}$) and 
measures the fraction of variance in $f$ that can be attributed to each parameter in the group $\ub$ and to the interactions between the parameters in the group. 
Note that if the cardinality of the group is $1$, \eqref{eq:grouped_main_effect} reduces to the traditional first-order Sobol' index for a single input.

The \textit{grouped total Sobol' index} (or \textit{grouped total effect index}) for the group of parameters $\ub$ is defined as
\begin{align}
    T_\ub = \frac{\E_{\X_{\sim \ub}}\left[\Var_{\X_\ub}(f(\X)|\X_{\sim \ub})\right]}{\Var(f(\X))}
    = 1 - \frac{\Var_{\X_{\sim \ub}} \left( \E_{\X_{\ub}} \left[ f(\X) | \X_{\sim \ub}\right] \right) }{\Var(f(\X))}.
   \label{eq:grouped_total_effect}
\end{align}
The total Sobol' index represents the contribution to the total variance due to $\X_{\ub}$, including interactions with other model parameters.
As with the grouped first-order Sobol' index, if the cardinality of group $\ub$ is $1$,~\eqref{eq:grouped_total_effect} reduces to the traditional total Sobol' index for a single input.

%% file: sobol_est.tex
\subsubsection{Computation of grouped Sobol' indices}
\label{sec:comp_sobol_indices}

\firstline{We adopt the most common approach to compute Sobol' indices, the \textit{pick-freeze} method.}
We employ the estimators presented in~\cite{prieur_2017,Owen_2013} for the main and total effect indices defined in~\Cref{sec:grouped_indices}. 
Derivations for the estimator of the numerator in the grouped first-order index can be found in~\cite{prieur_2017}, while Appendix A presents a similar derivation corresponding to the grouped total index.
Let $\pi({\bm x})$ be the probability density function (PDF) associated with $\binpRV$, and let ${\bm x},{\bm x}' \sim \pi$ denote replicate random vectors whose PDF is $\pi$, i.e., $\pi({\bm x}) \equiv \pi({\bm x}')$.
The numerator for the grouped first-order Sobol' index is computed as
\begin{align}
    \Var_{\X_\ub}\left(\E_{\X_{\sim \ub}}\left[f(\X)|\X_\ub\right]\right) &\approx
    \frac{1}{N}\sum_{i=1}^N f(\x^{(i)})\left( f(\x'^{(i)}_{\sim \ub}, \x^{(i)}_\ub) - f(\x'^{(i)}) \right), \quad \x^{(i)}, \x'^{(i)} \sim \pi(\x).
\end{align}
The numerator of the grouped total Sobol' index is computed as 
\begin{align}
    \E_{\binpRV_{\sim \ug}} \left[ \V_{\binpRV_{\ug}} \left( f\left(\binpRV \right) | \binpRV_{\sim \ug} \right) \right] \approx
\frac{1}{2N} \sum_{i=1}^N \left( f(\x^{(i)}) - f(\x^{(i)}_{\sim \ub},\x'^{(i)}_\ub) \right)^2, \quad \x^{(i)}, \x'^{(i)} \sim \pi(\x).
\end{align}
Finally, the total variance appearing in the denominators of the indices is computed as
\begin{align}
    V(f(\X)) &\approx \frac{1}{2N}\sum_{i=1}^N \left[\left(f(\x^{(i)}) - \hat{\mu}_f \right)^2 + \left(f(\x'^{(i)}) - \hat{\mu}'_f \right)^2\right], \\
    \hat{\mu}_f &\approx \frac{1}{N}\sum_{i=1}^N f(\x^{(i)}), \quad
    \hat{\mu}'_f \approx \frac{1}{N}\sum_{i=1}^N f(\x'^{(i)}). \nonumber
\end{align}

\subsubsection{Treatment of hierarchical distributions}\label{sec:hier_sobol}
\firstline{It is not immediately obvious how to appropriately incorporate hierarchical uncertainty representations for MFU parameters (as defined in~\eqref{eq:hierarchical_distribution}) into global sensitivity analysis.}
Previous works have considered differing interpretations of Sobol' indices in the context of hierarchical uncertainty characterizations.
The work \cite{mandel_randomized_2018}
formulates \textit{randomized} Sobol indices, which compute an average (with respect to the hyperparameter distribution) Sobol' index.
This approach is based on the interpretation that model parameter uncertainty is reducible (epistemic), while randomness in the hyperparameters reflects variability across experiments (irreducible).
Thus, the authors intend for the standard Sobol' index to reflect sensitivity for a particular experiment, with the \textit{randomized} indices representing aggregate sensitivity over all experiments.
Similarly, the work of~\cite{krzykacz-hausmann_approximate_2006} proposes that for computer models with both reducible and irreducible sources of uncertainty, one should compute an expected Sobol' index; however, alternative to~\cite{mandel_randomized_2018}, the expectation is with respect to the epistemic variables.

\firstline{The interpretation of Sobol' indices in both~\cite{mandel_randomized_2018, krzykacz-hausmann_approximate_2006} fundamentally differs from that underpinning our methodology.}
Here we wish to characterize the aggregate influence of MFU parameters and hyperparameters on model outputs.
We thus define the MFU group in terms of the MFU parameters and their hyperparameters: $\X_{\ug} = (\theta, \phi)$.
All estimators presented in~\Cref{sec:comp_sobol_indices} and derived in~\Cref{sec:Appendix_A} hold for samples drawn from a hierarchical posterior since the estimators only require sampling from the joint distribution over the group and make no assumptions about the structure of the distribution.
This choice to represent the joint parameters and hyperparameters as a group ensures the resulting Sobol' index reflects uncertainty associated with the MFU representation as a whole.

\subsubsection{Treatment of dependent groups and inputs}\label{sec:dependent_groups_SA}
\firstline{In practice, the estimation of the MFU parameters occurs alongside other model parameters.}
Thus, the group of MFU parameters may be correlated to the other model parameters in the posterior distribution. 
This violates the assumption of statistical independence between the groups. Generalizations of Sobol' indices in the presence of dependent inputs are considered in~\cite{10.1214/12-EJS749,MARA2012115,KUCHERENKO2012937,RABITZ20107587}.
However, the traditional interpretation of Sobol' indices attributing a relative variance contributing does not apply when the inputs are dependent, e.g., when they are correlated.
Various alternative global sensitivity analysis approaches have been considered~\cite{ZHOU20144885,Iooss_2019,DaVeiga03052015}.

In this article, we focus on the total Sobol' indices thanks to their alternative approximation theoretic interpretation that generalizes to the case with dependent inputs~\cite{Hart_2018}. 
Specifically, the total Sobol' index for the group of inputs $\X_{\ug}$ corresponds to the relative approximation error if $f$ is optimally approximated (in the $L^2$ sense) by a function that only depends on $\X_{\sim \ug}$. 
From this perspective, a small total index $T_{\ug}$ implies that $f$ can be approximated well without information from $\X_{\ug}$, and hence $\X_{\ug}$ is not influential. 

The presence of input correlations typically results in smaller total indices because some information about parameter variability may be captured through correlations with other parameters.
In the context of comparing two groups, model parameters and MFU parameters, a larger total index implies that information about variability from one group cannot be captured through correlations with the other group, and hence contributes more significantly to uncertainty in the function's output. 
Therefore, if there are only two groups of parameters, it is possible to reason about the relative importance of the groups based on their total effects indices.
However, if there are more than two groups of input parameters, this reasoning breaks down, since it is possible that a group has a large Sobol' index because it is strongly correlated with a group that is highly influential, rather than because it is influential in its own right. 

It would not be meaningful to directly compare the magnitudes of total effects indices between prior and posterior distributions---the dependence structure between the input groups has changed, and the posterior predictive variance is typically smaller than the prior predictive variance. 
However, total effect indices' numerators can be meaningfully compared since the indices for prior and posterior both correspond to absolute function approximation errors with the same units. 
The reduction in magnitude from prior to posterior indicates the decrease in variability due to calibration. 
By comparing the  total effect index numerators for model parameters and MFU parameters, we may understand how the relative importance of the two groups changed from prior to posterior. 
This is only possible, however, because we have split the parameters into two groups, so that there is no confounding dependence on a third group of parameters.
~\hyperref[ex:sobol]{Example~2.3} demonstrates computation and interpretation of these indices for the MFU representation introduced in Example~\ref{ex:hier}.

\begin{examplebox}
\label{ex:sobol}
The model and MFU parameters exhibit significant dependence after Bayesian calibration, as shown in~\Cref{fig:polynomial_posterior_ccs}.

\begin{center}
  \includegraphics{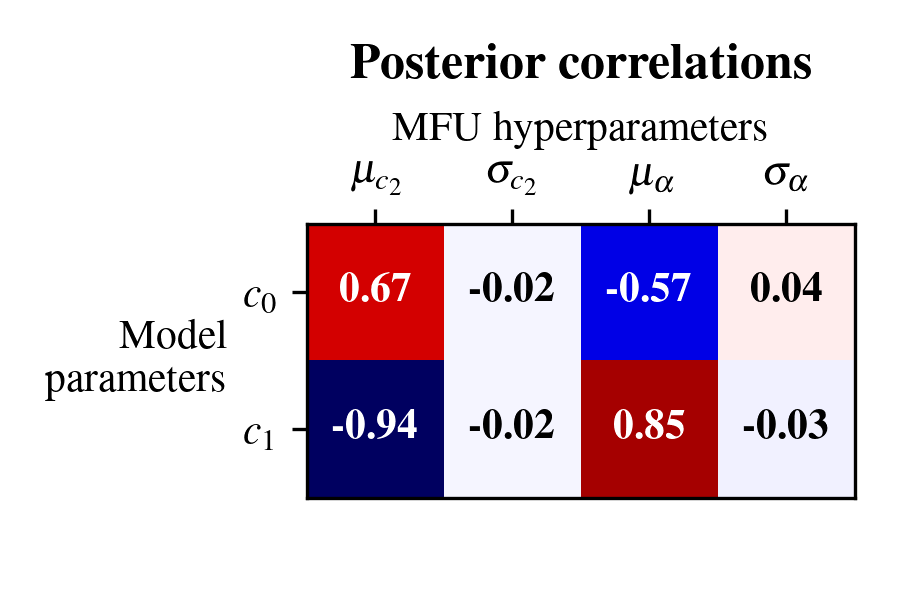}\vspace{-1em}
  \captionof{figure}{Correlation coefficients between model parameters (rows) and MFU hyperparameters (columns) after Bayesian calibration.}
  \label{fig:polynomial_posterior_ccs}
\end{center}

Nevertheless, employing the previously discussed interpretation, we can compute the total effects indices for the model and MFU parameter groups. 
\Cref{fig:polynomial_sobols} illustrates the total indices for both for the prior distribution and the posterior distribution, where one can see that the total effect index for MFU parameters exceeds 1, which would not be possible if the parameter groups were independent. 

\begin{center}
  \includegraphics{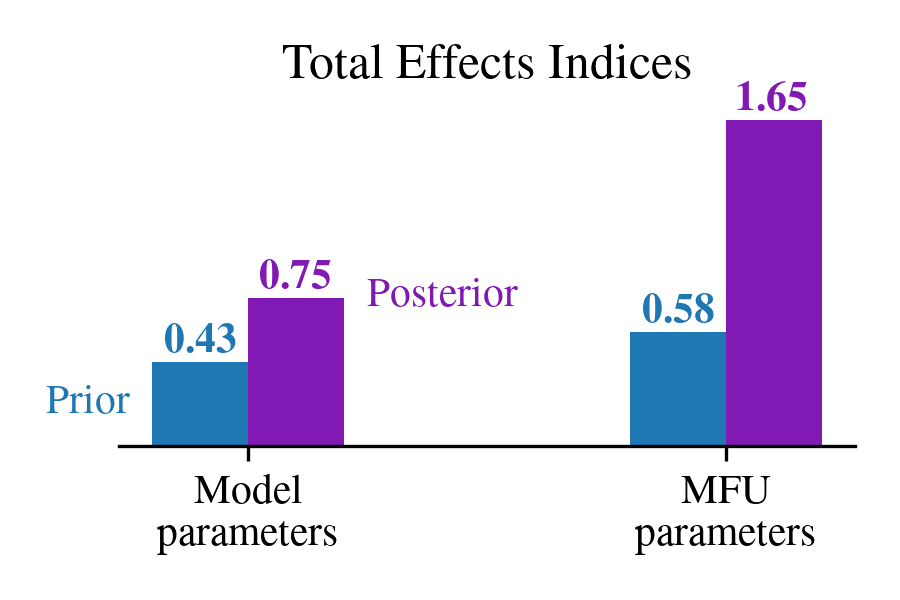}\vspace{-1em}
  \captionof{figure}{Total effects indices of $\widetilde{f}(2)$ for input distributions prior to calibration (blue) and after calibration (purple). }
  \label{fig:polynomial_sobols}
\end{center}

As discussed previously, it can be more meaningful to compare the total index numerators when comparing prior and posterior sensitivities, which are shown in~\Cref{fig:polynomial_numerators_x} before and after Bayesian calibration.
We observe that prior to calibration, MFU parameters aren't influential to the output except for the largest $x$ values.
After calibration, they are equally influential to, or greater than, the influence of model parameters for all $x$ values.
Additionally, we observe that the magnitude of the numerator has reduced significantly after calibration due to the overall reduction in output variance post-calibration.
This illustrative example shows how Bayesian calibration can alter the relative importance of parameters in the model.

\begin{center}
  \includegraphics{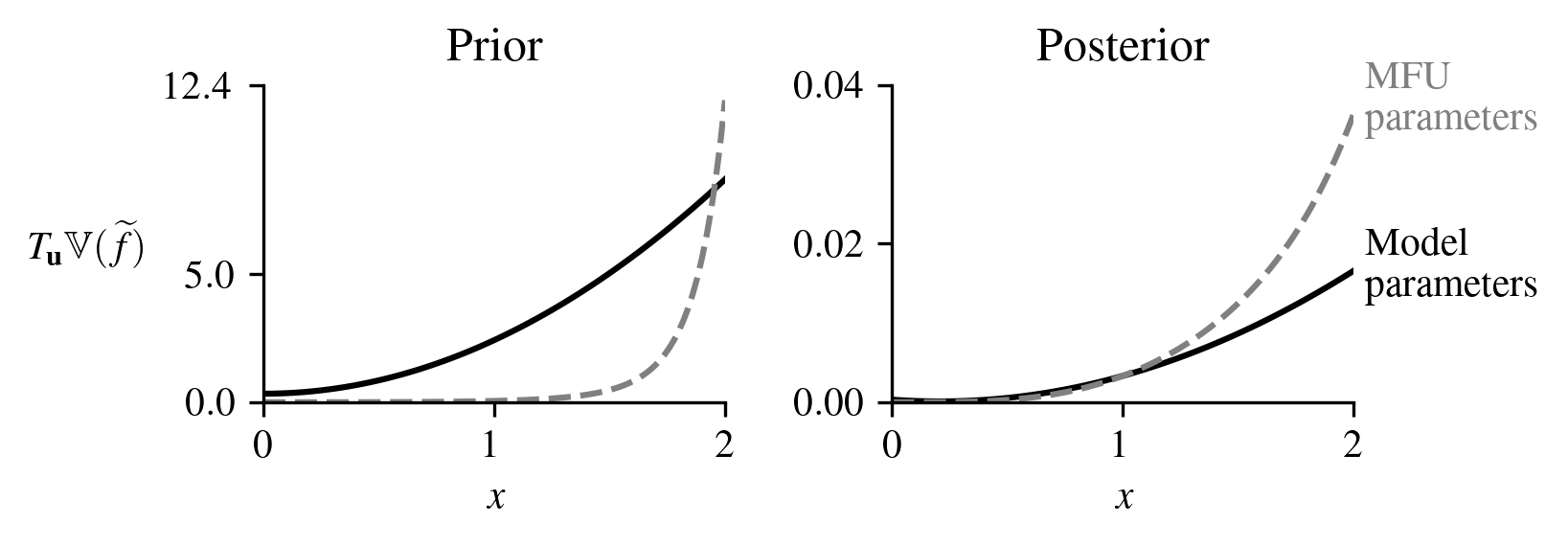}\vspace{-1em}
  \captionof{figure}{Numerators of total effects indices of $\widetilde{f}(x)$ for model parameters (solid) and MFU parameters (dashed), for prior (left) and posterior (right) distributions.}
  \label{fig:polynomial_numerators_x}
\end{center}

\end{examplebox}

%% file: results.tex
\section{Results} \label{sec:results}

\firstline{Here we present theoretical and numerical results.}
In~\Cref{sec:robustness} we provide theoretical proof of robustness of grouped sensitivity analysis to MFU parameterization.
In~\Cref{sec:cont_transport} we present numerical results demonstrating our method and corroborating our theoretical result. 
Our numerical studies are carried out in the context of a model for contaminant transport through a heterogeneous porous medium.
Code to reproduce all numerical experiments in this section is publicly available at \url{https://github.com/sandialabs/MFUQ-MFU-GSA}.

\input{robust.tex}

\input{cont_tran_ex}
\input{cont_tran_sobol}
\input{cont_tran_hier}

%% file: robust.tex
\subsection{Theoretical result: robustness of grouped sensitivity analysis to MFU parameterization}
\label{sec:robustness}
\firstline{Since multiple MFU representations may be valid in a modeling scenario,
it is critical to understand how differences in MFU representation choice impact resulting Sobol' index computation.}
Here, we provide the assumptions under which the differences in Sobol' indices for varying MFU representations will be bounded. 

Let us consider two mathematical models, $f:\Omega_{\X} \to \mathbb{R}$ and $q:\Omega_{\Xtilde} \to \mathbb{R}$, whose difference lies in their MFU representation.
The inputs for $f$ are given as $\X = \{\X_\ub, \X_\vb\}$, where $\X_\ub$ represents the MFU parameters. 
Similarly, the inputs for $q$ are given as $\Xtilde = \{\Xtilde_\ubtilde, \X_\vb\}$, where $\Xtilde_\ubtilde$ represents the alternative MFU parameters. 
Note that parameters other than MFU parameters are the same between both models and given as $\X_\vb$.
The Sobol' indices for groups $\ub$ and $\ubtilde$ are 
\begin{align*}
    S_\ub^g = \frac{\Var_{\Xu}\left(\E_{\X_\vb}[f(\X)|\Xu]\right)}{\Var(f(\X))}
    \quad \quad
    S_\ubtilde^g = \frac{\Var_{\Xtilde_\ubtilde}\left( \E_{\X_\vb}[q(\Xtilde) | \Xtilde_\ubtilde]\right)}{\Var\left(q(\Xtilde)\right)},
\end{align*}
and the grouped total Sobol' indices are
\begin{align*}
    T_\ub = \frac{\E_{\Xu}\left[\Var_{\X_\vb}\left(f(\X)|\Xu\right)\right]}{\Var(f(\X))}
    \quad \quad
    T_\ubtilde = \frac{\E_{\Xtilde_\ubtilde}\left[ \Var_{\X_\vb}\left(q(\Xtilde) | \Xtilde_\ubtilde\right)\right]}{\Var\left(q(\Xtilde)\right)}.
\end{align*}
We assume statistical independence between parameter groups, i.e.,~$\pi(\X_\ub, \X_\vb)=\pi(\X_\ub)\pi(\X_\vb)$.
\begin{prop}\label{prop:bounds}
\normalfont
If 
\begin{enumerate}
    \item $\norm{\Var_{\X_\ub} \left( f(\X)| \X_\vb \right) - \Var_{\Xtilde_{\ubtilde}} \left( q(\Xtilde)| \X_\vb \right) }_{\mathcal{L}_1} < \errone$. 
    \item $| \V(f(\X)) - \V(q(\Xtilde))| < \errtwo$. 
    \item $\V(f(\X)) = 1$. 
\end{enumerate}
Then 
$\big| S_\ub^g - S_\ubtilde^g \big| < \errone + 2\errtwo$.
\end{prop}

\begin{prop}\label{prop:bounds_2}
    \normalfont
    If Assumptions (1)-(3) from Proposition~\ref{prop:bounds} hold,
    then 
    $\left|\ts_{\ub} - \ts_{\ubtilde} \right| < \errone + \errtwo$. 
\end{prop}

Assumptions (1) and (2) are motivated by the belief that two valid MFU representations should have a similar effect in the output statistics of the QoI prediction. 
Assumption (3) is for mathematical convenience and physically immaterial as it corresponds to a change in units (for which Sobol' indices are invariant).
Note that both $f$ and $q$ are rescaled in the same manner to maintain consistency in units.
For a proof of Propositions~\ref{prop:bounds} and~\ref{prop:bounds_2}, see Appendix~\ref{sec:Appendix_B}.

%% file: cont_tran_ex.tex
\subsection{Numerical results: contaminant transport through a heterogeneous porous medium}
\label{sec:cont_transport}

\firstline{First, we present the application problem description.}
Transport through porous media is governed by the advection-diffusion equation, where the flow field (fluid velocity) depends on the porous medium's permeability and porosity fields through Darcy's law~\cite{bear_modeling_2010}.
Limitations in sensing technology make it impossible to perfectly characterize the heterogeneous properties of the subsurface. 
While it is possible to generate random instantiations of the subsurface and collect summary statistics to address missing information about the subsurface properties, this can be computationally prohibitive. 
Instead, it is common to assume statistical homogeneity in the subsurface (i.e.,~the statistical properties do not change throughout the domain) and to average the equations with the aim of predicting mean behavior.

\firstline{We consider the averaged governing equation for the mean contaminant concentration $\mean{c}$:}
\begin{equation}
  \begin{aligned}
    \diffp[]{\meanc}{t}(x,t) + \mean{u}\diffp[]{\meanc}{x}(x,t) + \diffp[]{\mean{u'c'}}{x}(x,t) = \nu_p \diffp[2]{\mean{c}}{x}(x,t), \\
     \mean{c}(0,t) = \mean{c}(L_x,t), \\
     \mean{c}(x,0) = \exp\left(-\frac{(x - s)^2}{2\ell^2}\right),
   \end{aligned}
   \quad
   \begin{aligned}
    x \in [0,L_x],
   \end{aligned}
   \label{eq:averaged}
\end{equation}
where $\mean{u}$ is the mean velocity, $\partial_x \mean{u'c'}$ is the \textit{dispersion} resulting from fluctuations in the velocity field, and $\nu_p \partial_{xx} \meanc$ is the pore-scale diffusion of the contaminant.
The initial condition is a Gaussian pulse with width $\ell=0.1$ and mode $s$.
We assume periodic boundary conditions, which is valid provided the velocity fluctuations are homogeneous with correlation lengths small compared to $L_x=4$.

\firstline{\Cref{eq:averaged} cannot be solved in this form because the dispersion term depends on the fluctuations of $u$ and $c$, which are not resolved in the upscaled model.}
An assumption about the mathematical form of the dispersion term must be introduced---this would be represented by the $a$ in our definition of the governing equations in~\eqref{eq:governing_eq}.
For heterogeneous porous media, significant fluctuations in the velocity lead to nonlocal transport of the contaminant, often called \textit{anomalous diffusion}~\cite{levy_measurement_2003,yeh_flow_2015,neuman_perspective_2009}.
This effect is shown in~\Cref{fig:anomalous_diffusion_example}, which compares the velocity fluctuations and evolution of the Gaussian initial condition in~\eqref{eq:averaged} for a homogeneous and a heterogeneous porous medium.
The higher-magnitude and larger-lengthscale fluctuations of the velocity in the heterogeneous case transports the contaminant downstream at drastically different rates across the domain, resulting in larger concentrations downstream.
In contrast to the homogeneous case, this effect is considered ``anomalous'' since its effect on the concentration does not mimic typical diffusion.
\begin{figure}[h!]
    \centering
    \includegraphics{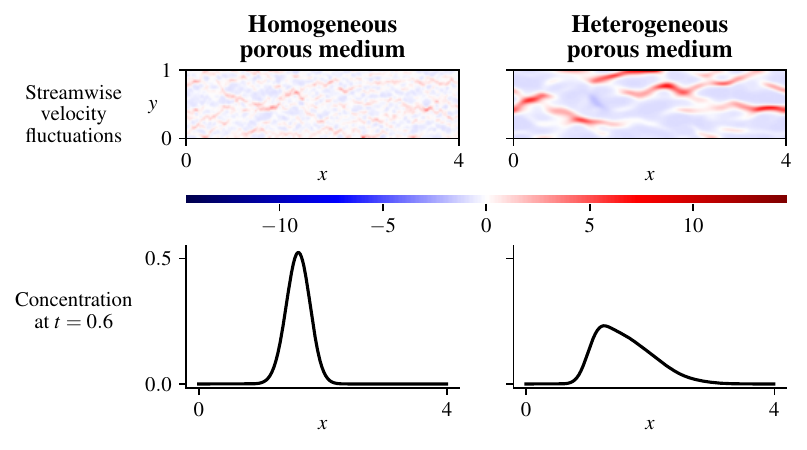}\vspace{-1em}
    \caption{A comparison of streamwise velocity fluctuations and concentration fields for a homogeneous porous medium (left) and a heterogeneous porous medium (right). Velocity fluctuations are for a single realization of a velocity field, while the concentration is averaged over many realizations and depthwise ($y$-direction) averaged.}
    \label{fig:anomalous_diffusion_example}
\end{figure}

\firstline{The mean contaminant concentration was computed for~\Cref{fig:anomalous_diffusion_example} using direct numerical simulation, by generating many porous medium realizations and solving the detailed governing equations, then averaging the solutions.}
However, in general, the exact nature of the nonlocality in the dispersion is uncertain due to lack of detailed information about the subsurface.
This leads to uncertainty in the assumed form of the dispersion.

\subsubsection{Potential model-form uncertainty representations for dispersion}\label{sec:mfu_reps}
\firstline{In this work we focus on MFU representations given as linear operators acting on $\mean{c}$. }
The most general form of such an operator is denoted as
\begin{align}
  \diffp{\mean{u'c'}}{x} &\approx -\Lcal \mean{c}. \label{eq:general_linear}
\end{align}
Under the mild assumption that $\Lcal$ is shift invariant (does not depend on absolute location, just relative distances), the Fourier modes are the eigenfunctions of the operator $\Lcal$:
\begin{align*}
    \Lcal e^{i a_k} = \lambda_k e^{i a_k}, \quad a_k = \frac{2 \pi k }{L_x}, \quad k\in \mathbb{Z}.
\end{align*}
\Cref{eq:averaged} thus admits a Fourier series solution, and the action of $\Lcal$ on $\mean{c}$ can be entirely described in terms of its eigenvalues $\lambda_k$.
The eigenvalues may be complex-valued, with real parts influencing the magnitudes and imaginary parts influencing the phases of associated Fourier coefficients~\cite{portone_bayesian_2022}.
In this work, we consider three MFU representations for dispersion. 
First, a general linear operator parametrized by its eigenvalues, with as many eigenvalues as there are terms in the Fourier series solution:
\begin{align}
    \Lcal_g [e^{ia_k}] &= \lambda_k e^{i a_k}, \quad k=1, \ldots, N_k. \label{eq:general_linear}
\end{align}
Second, a Reisz fractional derivative operator, parameterized by its scaling coefficient $\nu_m$ and fractional power $\alpha$:
\begin{align}
  \Lcal_{f} [e^{i a_k }] &= -\nu_m\diffp[\alpha]{}{x} [e^{i a_k} ] = - \nu_m (i a_k)^\alpha e^{i a_k}. \label{eq:fractional} 
\end{align}
Finally, a complex operator inspired by the Reisz fractional derivative operator, parameterized by the complex and imaginary counterparts of its eigenvalues:
\begin{align}
    \Lcal_{c} [e^{i a_k}] &= \left( - \nu_m^r (a_k)^{\alpha_r} + i \nu_m^i (a_k)^{\alpha_i} \right) e^{i a_k}. \label{eq:complex_fractional}
\end{align}

\firstline{Each MFU representation considered here has potential strengths and weaknesses.}
The general linear operator given by~\eqref{eq:general_linear} is the most flexible, since each eigenvalue can vary independent of the others. 
However, the number of parameters that must be constrained and informed by data could be  large for a fine spatial discretization.
A parametrization with dimensionality in the hundreds could easily result in computational intractability. 
Fractional derivative operators like~\eqref{eq:fractional} yield nonlocal behavior for fractional powers $\alpha \in (1,2)$, making them an attractive option to represent MFU in dispersion for this application problem. 
Their simple parametrization in terms of the fractional power and scaling coefficient also makes them more tractable than the general linear operator;
however, they may not be flexible enough to reproduce dispersion effects in all problems---for example, in practice the real and imaginary parts may need to vary independently to represent differing influences on diffusion and advection by dispersion.
This inflexibility motivated the form of the complex operator inspired by the fractional derivative given in~\eqref{eq:complex_fractional}. 
While having relatively few parameters, it is more flexible in the behavior of the real and imaginary parts of the eigenvalues.

\firstline{Each MFU representation can be constrained to respect physical principles.}
For example, for the general linear operator $\Re[\lambda_k] < 0$ ensures diffusive behavior in $\mean{c}$, and $\lambda_0=0$ conserves mass.
For additional constraints that can be placed on the parametrization of the general linear operator, see~\cite{portone_bayesian_2022}.
For the fractional derivative and complex fractional operators, all scaling coefficients must be positive to ensure diffusive behavior (for real parts) and advection downstream (for imaginary parts).
Additionally, the fractional powers must remain in the range $(1,2)$ to express nonlocality of the dispersion.
A benefit specific to the fractional derivative operator is that it preserves positivity of $\mean{c}$ by construction, which is not as easily imposed on the other operators. 

%% file: cont_tran_sobol.tex
\subsubsection{Model-form uncertainty propagation without data}\label{sec:MFU_forward}
\firstline{Here we demonstrate how MFU can be characterized and propagated without access to calibration data by bringing to bear physical properties and knowledge about the modeled phenomenon.}
In all following discussions, the quantity of interest is the mean concentration $\mean{c}$ at the outflow boundary at time $t=1.5$.
We first characterize uncertainty in the model parameters bulk velocity $\mean{u}$,  pore-scale diffusion coefficient $\nu_p$, and contaminant source location $s$:
\begin{equation}
\begin{aligned}
  \mean{u} &\sim  1 \cdot \left( 1 + \mathcal{U}[ -0.1, 0.1]\right),\\
  \nu_p &\sim  0.01 \cdot \left(1 + \mathcal{U}[-0.2, 0.2 ] \right), \\
  s &\sim \mathcal{U}[0.2, 1.5].
\end{aligned}
\label{eq:other_param_distributions}
\end{equation}
The distributions for $\mean{u}$ and $\nu_p$ represent a $\pm10$ and $20$ percent uniform uncertainty, respectively, about a nominal value.
The distribution for $s$ reflects the assumption that the source location is in the upstream section of the computational domain without being too close to the upstream boundary.

\firstline{Here, we represent MFU in the dispersion term using the fractional derivative operator~\eqref{eq:fractional} for its simple parametrization and because it preserves positivity in $\mean{c}$ by construction---an advantage in this case since we don't have calibration data to inform the more general parameterizations discussed above.}
Encoding physical properties (see~\Cref{sec:mfu_reps}) into the probabilistic representation of the MFU parameters of~\eqref{eq:fractional} provides the following assumed distributions:
\begin{align}
  \label{eq:nu_dist} \nu_m &\sim \scriptU[0.05, 0.15], \\
 \label{eq:alpha_dist} \alpha &\sim \text{Triangular}[1,2,1.5],
\end{align}
where~\eqref{eq:nu_dist} is based on the assumptions that dispersion dominates pore-scale diffusion (and thus the scaling coefficient $\nu_m$ should be larger than $\nu_p$) and that dispersion is not so great that it completely diffuses the solution by the time it has been transported the length of the domain.
The expression given by~\eqref{eq:alpha_dist} is based on the knowledge that the fractional power should fall in the range $[1,2]$, with the mode set at $1.5$ to weakly concentrate fractional powers away from the boundary values, which result in pure advection or pure diffusion.

Forward propagation of these uncertainties is shown in~\Cref{fig:FRADE_forward_UQ}, where~\Cref{fig:FRADE_spatial_samples} displays 10 sample evolutions of $\mean{c}$ over the computational domain. 
Note the varying degrees of nonlocality in the concentration field, as evidenced by some samples being right-skewed with a heavier tail downstream. 
A histogram of our QoI, the outflow concentration at $t=1.5$, is shown in~\Cref{fig:FRADE_qoi_samples}.
Note that all samples are positive due to the fractional derivative MFU representation enforcing positivity of concentrations by construction. 

\begin{figure}[h]
    \centering
    \begin{subfigure}{0.45\textwidth}
    \centering
\includegraphics{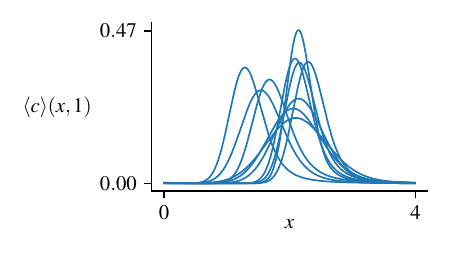}\vspace{-1em}
    \caption{10 sample evolutions of the mean concentration at $t=1$ for different input sample values.}
    \label{fig:FRADE_spatial_samples}
    \end{subfigure}
    \hspace{2em}
    \begin{subfigure}{0.45\textwidth}
        \centering
    \includegraphics{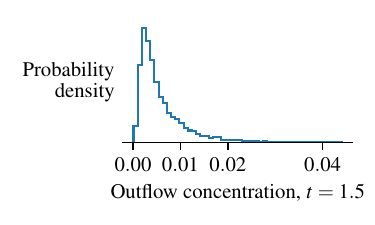}\vspace{-1em}
    \caption{Histogram of the QoI, outflow concentration at $t=1.5$, computed for 1000 input sample values. The $x$ ticks indicate lower bound, mean, 95$^{th}$ percentile, and upper bound from left to right.}
\label{fig:FRADE_qoi_samples}
    \end{subfigure}
    \caption{Output samples generated according to input probability distributions~\Cref{eq:other_param_distributions,eq:nu_dist,eq:alpha_dist}.}
    \label{fig:FRADE_forward_UQ}
\end{figure}

\subsubsection{Verification of robustness of Sobol' indices for different MFU parametrizations}\label{sec:sobol_comp}

\firstline{Here we numerically corroborate the theoretical results presented in~\Cref{sec:robustness}, which state that two different MFU parameterizations inducing similar output statistics should result in similar Sobol' indices.}
Since the fractional derivative is a specialization of the general linear operator, we can derive a parametrization of $\Lcal_g$ that produces very similar output statistics to those induced by $\Lcal_f$, which enables us to evaluate the theoretical results numerically.
Following the notation of~\Cref{sec:robustness}, we compare the grouped Sobol' index for the fractional derivative MFU parameters $\X_\ub=[\nu_m,\alpha]$ and the general linear operator MFU parameters $\Xtilde_\ubtilde=\lb$ relative to the other uncertain model parameters in \Cref{eq:averaged}, $\X_{\vg} = [\langle u\rangle, \nu_p, s]$.
We additionally denote the QoI obtained with the fractional derivative MFU representation as $f(\X)$ and with the general linear operator as $g(\Xtilde)$.
We adopt the probability distributions for $\X_\ub$ and $\X_{\vg}$ discussed in~\eqref{sec:MFU_forward}.

\firstline{To obtain a general linear operator MFU representation that produces similar statistics to the fractional derivative representation, we use data-consistent stochastic inversion (DCI)~\cite{butler2018,butler2020}.}
DCI is a measure-theoretic framework for solving stochastic inverse problems that identifies parameter distributions whose predicted outputs reproduce distribution on the observed data. 
Therefore, DCI is leveraged to determine
a distribution on the eigenvalues $\lambda_k$ that results in the the distribution---and hence statistics--of $g(\X)$ being consistent with$f(\X)$ (the QoI corresponding to  the fractional derivative representation).
More details on this procedure can be found in~\Cref{sec:Appendix_B}.

\firstline{Numerical comparison of the Sobol' indices for the fractional derivative and generalized linear operators are carried out as follows.}
50 replicate Sobol' indices are computed using 5e4 independent samples.
For fairness in comparison, the same set of random samples of the auxiliary parameters $\langle u\rangle, \nu_p$ and $s$ are used to estimate indices across the two MFU representations.
For each replicate sample, the QoI is rescaled so that its variance for the fractional derivative MFU representation is equal to $1$. 
The mean variance over replicates of the QoI for the general linear operator representation is $1.02$ with a standard deviation of $0.01$,
indicating an approximately 2\% difference in the variance from the fractional derivative parameterization.

\firstline{A comparison of the main and total effect Sobol' indices for the two MFU representations is presented in \Cref{fig:sobols_compared}.}
The Sobol' index values are qualitatively similar; the conclusions would be identical for a ranking problem (ranking the most important sources of uncertainty).
\begin{figure}[h]
  \centering
  \includegraphics{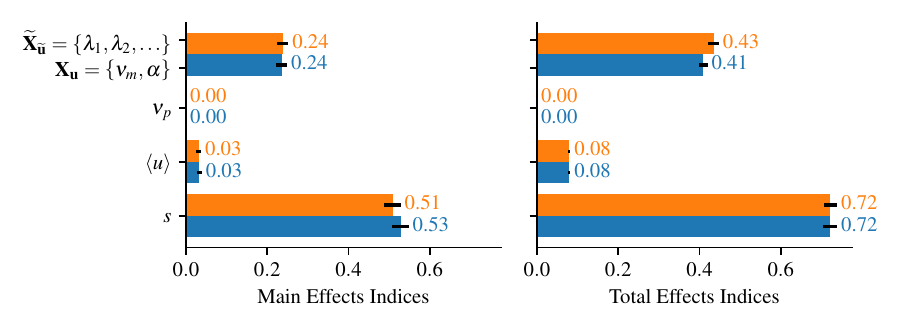}\vspace{-1em}
  \caption{A comparison of 50 replicate Sobol' indices for the fractional derivative MFU representation case (orange) and the general linear operator (blue). Bar charts and numeric values are the mean over the replicates. Black error bars represent a $2\sigma$ bound over replicates.}
  \label{fig:sobols_compared}
\end{figure}
\firstline{Having numerically corroborated our theoretical findings, we now test the tightness of the bounds derived in~\Cref{sec:robustness}.}
As discussed above, the norm of differences in total variances $\errtwo$ is 0.02 on average.
We compute the norm of differences in conditional variances $\errone$ with quadrature when possible and Monte Carlo sampling otherwise, again computing 50 replicate samples.
On average, $\errone=0.12$ with a standard deviation of $0.02$.
The resulting bounds on the differences between the main and total effects indices for the MFU representation parameters are therefore approximately $0.16$ and $0.14$, respectively.
The results in~\Cref{fig:sobols_compared} indicate the actual difference between the main and total effects indices to be $0.008$ and $0.025$ on average, which could indicate the bounds derived in \Cref{sec:robustness} are not tight.

Future work could aim to derive tighter bounds;
however, for the purpose of this work, this numerical example is sufficient to confirm the intuitive understanding that MFU parametrizations resulting in similar QoI statistics should have similar Sobol' indices.

%% file: cont_tran_hier.tex
\subsubsection{Calibrating model-form uncertainty with data}\label{sec:MFU_inverse}
\firstline{Here we demonstrate how the hierarchical MFU formulation can be informed using data and show how sensitivities before and after Bayesian inference can be compared.}
This work serves as an application of the methods demonstrated on a simple polynomial example in~\Cref{sec:back_methods}, applied to the more complex application problem of contaminant transport.
We found in preliminary studies that the fractional derivative MFU representation defined in~\eqref{eq:fractional} was not adequately flexible to capture MFU for the calibration case, so we use the complex fractional MFU representation defined in~\eqref{eq:complex_fractional} in these results.

\subsubsection*{Prior densities}
\firstline{First, we assign more flexible prior densities on the model parameters ($\mean{u}$, $\mean{u}$, and $s$) than the distributions assumed in~\Cref{sec:MFU_forward} to ensure artificially-imposed bounds do not influence the Bayesian inference.}
Since each of the model parameters is known to be positive with a specified nominal value, a log-normal prior density is assumed where the mode falls at the nominal value, and 95\% probability mass falls below 120\% of the nominal value.
Nominal values are 1, 0.01, and 1 for $\mean{u}, \nu_p,$ and $s$ respectively.

\firstline{Next, we consider the complex fractional MFU representation, where the scaling parameters $\nu_m^r$ and $\nu_m^i$ are known to be positive, and thus assumed to be log-normally distributed.}
We assign the same hyperpriors for both scaling parameters as there is no reason to distinguish between the real and imaginary parts of the eigenvalues
\textit{a priori}. 
Recall that the PDF for a log-normal random variable $x$ is defined as
\begin{align*}
    \pi(x) &= \frac{1}{x\sigma\sqrt{2\pi}}\exp\left( -\frac{(\ln(x)-\mu)^2}{2\sigma^2} \right).
\end{align*}
Therefore, hyperpriors must be defined for $\mu$ and $\sigma$.

\firstline{We do not have intuition for plausible values of the hyperparameters $\mu$ and $\sigma$, so we use intuition about plausible values for the scaling parameters to define the hyperpriors.}
Namely, we anticipate the scaling parameters fall in the range $[0.1,0.5]$, with weaker confidence in the lower vs.~the upper bound.
We express this belief by stating that $\mathbb{P}(\nu_m\leq 0.1)\approx 0.1$ and $\mathbb{P}(\nu_m\leq 0.5)\approx 0.99$. 
We map this belief about plausible values for the scaling parameters into nominal values for the hyperparameters $\mu$ and $\sigma$ using the quantile function for log-normal random variables, which returns a value $x$ associated with a probability level $p$:
\begin{align*}
    Q(p) = \exp\left( \mu + \sigma\Phi^{-1}(p)\right),
\end{align*}
where $\Phi$ is the CDF of the standard normal distribution.
Substituting $Q(0.1)=0.1$ and $Q(0.99)=0.5$ into the equation above, we arrive at two equations to solve for the two unknowns, $\mu$ and $\sigma$, which we will use as our nominal values $\mu_n$, $\sigma_n$.
Since $\mu$ needn't be positive, it is assumed to be a normal random variable with mean $\mu_n$ and standard deviation $0.5|\mu_n|$.
Since $\sigma$ must be positive, it is assumed to be log-normal with mode at $\sigma_n$ and 99\% probability below $1.5\sigma_n$.

\firstline{For the fractional powers, $\alpha^r$ and $\alpha^i$,  we again assume the the real and imaginary terms have the same prior specification.}
As in~\Cref{sec:MFU_forward}, we assume the fractional powers follow a triangular density on the range $[1,2]$, but the mode is inferred as a hyperparameter.
With no intuition about the mode beyond that it should fall within $[1,2]$, we assume a uniform prior on this range.

\subsubsection*{Likelihood}
\firstline{Calibration data is generated from the averaged governing equations for the mean contaminant concentration $\mean{c}$ assuming a general linear operator:}
\begin{equation}
    \begin{aligned}
        \diffp[]{\meanc}{t}(x,t) + \diffp[]{\mean{u'c'}}{x}(x,t) = \nu_p \diffp[2]{\mean{c}}{x}(x,t) + \Lcal\mean{c}, \\
     \mean{c}(0,t) = \mean{c}(L_x,t), \\
     \mean{c}(x,0) = \exp\left(-\frac{(x - s)^2}{2\ell^2}\right),
   \end{aligned}
   \quad
   \begin{aligned}
    x \in [0,L_x],
   \end{aligned}
   \label{eq:likelihood_model}
\end{equation}
with $\mean{u}=1.05$, $\nu_p=0.0095$, $s=0.9$, and $\ell, L_x$ defined identically to~\eqref{eq:averaged}. 
The linear operator is parameterized by its eigenvalues, where high-fidelity eigenvalues are computed via Monte Carlo sampling and ensemble averaging of the detailed governing equations, as described in~\cite{portone_representing_2019}.
The problem is discretized using a uniform grid with $N_x=512$.
Since the averaged equation admits a Fourier series solution, it can be evaluated at any point in time.

\firstline{To mimic the type of information typically available for this application area, we assume data can only be collected at a single well upstream of the outflow boundary.}
The calibration data is thus collected at $x=1.4$ with observations at $\mathbf{t}=[0.01, 0.02, \ldots, 0.2]$ (20 uniformly-spaced time observations up to $t=0.2$).
Because concentration data must always be positive, measurement error is assumed to be multiplicative with log-normal measurement error:
\begin{align*}
    d_i = \epsilon_m^{(i)} \cdot \mean{c}(1.4, t_i).
\end{align*}
We thus pose the likelihood in terms of the log-transformed concentrations, i.e., 
\begin{align}
    \log(d_i) = \log(\mean{c}(1.4, t_i)) + \log(\epsilon_m^{(i)}), \quad \log(\epsilon_m^{(i}))\sim\Ncal(0, (10^{-2})^2).
\end{align}
Note that this formulation of the measurement uncertainty is posed such that the median of the untransformed error distribution is 1.W

\firstline{Since the complex fractional MFU representation does not preserve positivity of $\mean{c}$ by construction, we penalize for negative concentrations in the likelihood.}
This is done in a rejection-sampling fashion, where if negative concentrations are observed for a given ordered pair of parameter values, a likelihood of $0$ is assigned. 
Since it would be computationally intractable to check for positivity over all times, we evaluate the model and check for negative concentrations at all spatial locations for times $0.1, 0.5, 1, 1.5,$ and $2$.
Note that since this check is only performed for a small number of time points during calibration, this approach can't penalize parameters that produce negative concentrations at other times. 
Therefore, forward uncertainty propagation is not guaranteed to result in positive concentrations for all model outputs.

\firstline{The ``true'' QoI value is also computed by solving~\eqref{eq:likelihood_model} for $\mean{c}$ at the outflow boundary ($x=4$) at $t=1.5$.}
Note that the times and location of calibration data differ significantly from the prediction QoI's, so proper calibration of uncertainties is critical to meaningfully extrapolate.

\subsubsection*{Inference results}
\firstline{
~\Cref{fig:cFRADE_model_param_posteriors} depicts the histograms of the prior and posterior densities for the model parameters vs.~their true values in the data-generating model.}
The source location for the contaminant, $s$, exhibits significant bias in the posterior relative to its true value, indicating that for this example, hierarchical calibration leveraging an MFU representation did not completely mitigate posterior bias and overconfidence.  
We expand upon challenges associated with hierarchical inference that impact its application to models leveraging an MFU representation in~\cite{portone_theoretical_2025}.
\begin{figure}[h]
    \centering
    \includegraphics{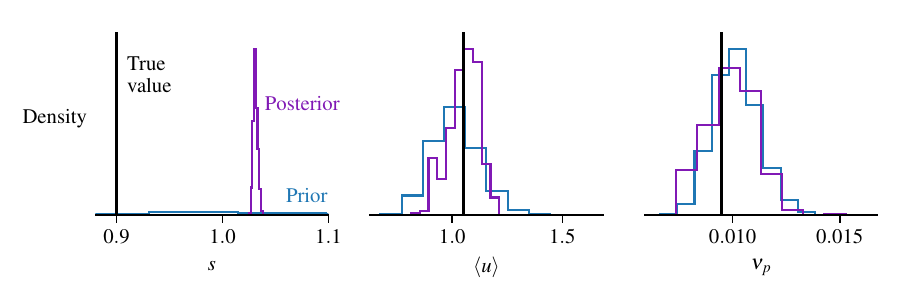}\vspace{-1em}
    \caption{Histograms of the prior and posterior densities for the model parameters vs.~their true values in the data-generating model.}
    \label{fig:cFRADE_model_param_posteriors}
\end{figure}

\firstline{Despite the observed bias in the model parameter posterior densities, the pushforward of the posterior probability densities to the prediction QoI (an outflow concentration significantly differing in time and location from the calibration observables), shown in~\Cref{fig:cFRADE_pushforwards}, shows a shift in the mode toward the true value of the QoI.}
Additionally, although the support of the prior and posterior pushforward densities does not visually differ significantly, the prior pushfoward has heavier tails that result in significantly higher variance ($10^{-3}$ and $3\cdot 10^{-4}$ for prior and posterior, respectively). 
\begin{figure}[h]
    \centering
    \includegraphics{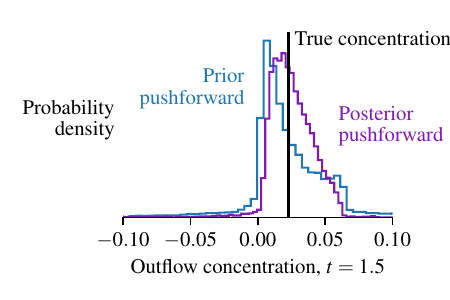}\vspace{-1em}
    \caption{Histograms of the prior and posterior pushforward to the QoI, compared to its true value.}
    \label{fig:cFRADE_pushforwards}
\end{figure}

\firstline{Negative concentrations are observed for both the prior and posterior pushforwards, although they are low probability.}
This is because the complex fractional MFU representation doesn't preserve positivity of concentration by construction.
However, due to the Bayesian calibration, the posterior pushforward has a lower incidence of negative values---2\% for posterior vs.~7\% for prior.

\subsubsection{Forward propagation of MFU post-calibration}
\firstline{Bayesian calibration resulted in correlations between model and MFU parameters in the posterior, as shown in~\Cref{fig:cFRADE_posterior_ccs} (further investigation would be needed to determine whether this correlation is physical or the result of identifiability issues).}
However, as discussed in~\Cref{sec:dependent_groups_SA}, it is still meaningful to compute the total effects indices for the two groups of parameters, as is depicted in~\Cref{fig:cFRADE_total_effects}.
Clearly, uncertainty in the dispersion assumption dominates uncertainty in the QoI post-calibration.
\begin{figure}[h]
    \centering
    \begin{subfigure}{0.48\textwidth}
    \centering
    \includegraphics[width=\textwidth]{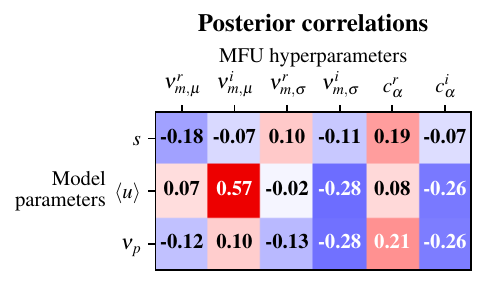}\vspace{-1em}
    \caption{Correlation coefficients between model parameters (rows) and MFU hyperparameters (columns).}
    \label{fig:cFRADE_posterior_ccs}
    \end{subfigure}\hspace{1em}
    \begin{subfigure}{0.48\textwidth}
        \centering
    \includegraphics{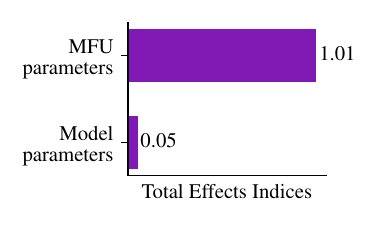}\vspace{-1em}
    \caption{Total effects indices for the QoI, computed using the joint posterior probability density for model and MFU parameters.}
    \label{fig:cFRADE_total_effects}
    \end{subfigure}
    \caption{Posterior density statistics.}
    \label{fig:posterior_corrs_and_total_effects}
\end{figure}

\firstline{To compare the relative influence of model and MFU parameters prior to and after calibration, we consider the unnormalized numerators of the total effects indices in~\Cref{fig:cFRADE_sobol_numerators}.}
The magnitude of the numerators for the posterior pushforward to the QoI has dropped relative to the prior pushforward. 
For both the prior and posterior, the MFU parameters appear to be more influential, although the ratio of the numerators between the groups becomes more pronounced post-calibration, indicating increased significance of MFU relative to model parameter uncertainty.
In the case of the posterior sensitivities, the ratio between MFU and model parameters is so great that the importance of model parameters could be considered negligible, indicating that further data collection or model improvement efforts should be focused on the dispersion assumption.
\begin{figure}[h]
    \centering
    \includegraphics{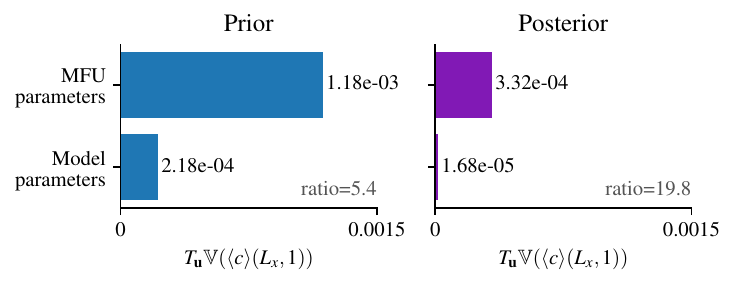}\vspace{-1em}
    \caption{Comparison of the numerators of total indices before (left) and after (right) Bayesian calibration.}
    \label{fig:cFRADE_sobol_numerators}
\end{figure}

%% file: conclusions.tex
\section{Conclusions}\label{sec:conclusions}

\firstline{This paper proposes a novel approach to quantify the impact of modeling assumptions on model predictions.}
We leverage model-form uncertainty (MFU) representations to represent uncertain modeling assumptions while grouped variance-based sensitivity analysis (VBSA) is used to quantify the impact of MFU on model predictions, which may be extrapolative.
This approach enables model assumptions and other sources of uncertainty to be holistically ranked in order of importance to model predictions, thereby facilitating more quantitative prioritization of modeling and data acquisition efforts towards the aspects of the model that most impact predictions.

\firstline{To establish robustness of the approach to MFU parameterization, we prove that for two MFU representations, the difference in Sobol' index values is bounded by statistical differences in their respective model outputs.}
Thus, MFU representations that produce ``similar'' effects on model outputs will have ``similar'' Sobol' index values.
We numerically demonstrate the robustness of the approach on an upscaled contaminant transport problem and find that for two MFU representations of dispersion, the difference in Sobol' indices is negligible.
Future work on the robustness of Sobol' indices to MFU representation will focus on the derivation of tighter bounds.

\firstline{We demonstrate how our approach can be applied with and without access to calibration data to inform the MFU representation.}
In the case where data is used to inform the representation via Bayesian inference, we show how VBSA can still be applied and interpreted despite violation of the common assumption that input distributions are statistically independent. 
We present numerical demonstrations in the form of a simple polynomial example as well as for the upscaled contaminant transport problem.
While the approach was applied to the hierarchical MFU formulation presented in~\Cref{eq:hierarchical_distribution}, it is amenable to other formulations, such as the one posed by~\cite{sargsyan_embedded_2019}, so long as the formulation yields an MFU parameterization.

\firstline{Embedding MFU representations within the governing equations, rather than augmenting individual model outputs with discrepancy terms, lets one measure the importance of model assumptions to model outputs over a range of model outputs, both observable and unobservable.}
This enables assessment of assumption importance to extrapolative model predictions and across a range of modeling use cases.
Future work could exploit this capability to measure the importance of assumptions to model outputs used in validation contexts to ensure that validation tests exhibit similar sensitivity to assumptions as those observed for target prediction use cases.
Application of our approach to assess validation test relevance would follow in the spirit of previous works which employ sensitivity analysis to assess the utility of validation tests~\cite{li_role_2016,paquette-rufiange_optimal_2023}---however, this work allows for the novel incorporation of MFU into such assessments.

\firstline{Additional work is needed to increase the practical applicability of our approach.}
A general approach to develop an MFU representation does not currently exist, and given the problem-specific nature of model assumptions, a general approach may be infeasible.
However, methods that enable rapid prototyping of MFU representations with minimal intrusion on simulation codes would significantly increase ease of adoption of our approach.
Additionally, pick-freeze methods to compute Sobol' indices are computationally costly, often requiring $\geq 10^5$ evaluations, with costs that scale linearly with the number of inputs/groups.
Methods that mitigate this cost are needed to increase tractability for computationally expensive models.
Given-data~\cite{li_efficient_2016,plischke_global_2013,borgonovo_common_2016} and multifidelity methods~\cite{qian_multifidelity_2018} for Sobol' index computation are promising avenues of investigation in this regard.
Finally, methods are needed for joint calibration of model and MFU parameters that mitigate confounding effects between the parameter groups.

%% file: acknowledgements.tex
\section*{Acknowledgements}
We thank Luis Damiano for helpful conversations on hierarchical Bayesian problems and Tian Yu Yen for helpful conversations on data-consistent inversion.
This work was supported by the Laboratory Directed Research and Development program
(Project 233072) at Sandia National Laboratories, a multimission laboratory managed and
operated by National Technology and Engineering Solutions of Sandia LLC, a wholly owned
subsidiary of Honeywell International Inc. for the U.S. Department of Energy’s National Nuclear
Security Administration under contract DE-NA0003525.
This article has been authored by an employee of National Technology \& Engineering Solutions of Sandia, LLC. The employee owns all right, title and interest in and to the article and is solely responsible for its contents. The United States Government retains and the publisher, by accepting the article for publication, acknowledges that the United States Government retains a non-exclusive, paid-up, irrevocable, world-wide license to publish or reproduce the published form of this article or allow others to do so, for United States Government purposes. The DOE will provide public access to these results of federally sponsored research in accordance with the DOE Public Access Plan \url{https://www.energy.gov/downloads/doe-public-access-plan}.
This paper describes objective technical results and analysis. Any subjective views or opinions that might be expressed in the paper do not necessarily represent the views of the U.S. Department of Energy or the United States Government.

%% file: Appendix_A.tex
\section{Grouped Sobol' index estimator derivations}
\label{sec:Appendix_A}

Let us first derive the estimator for the numerator for the grouped total Sobol' index, 
$$\E_{\binpRV_{\sim \ug}} \left[ \V_{\binpRV_{\ug}} \left( f\left(\binpRV \right) | \binpRV_{\sim \ug} \right) \right].$$
Let $\pi({\bm x})$ be the probability density function (PDF) associated with $\binpRV$, which is assumed to exist; similarly, $\pi({\bm x}_{\ug})$ and $\pi({\bm x}_{\sim \ug})$ are the joint PDFs associated with $\binpRV_{\ug}$ and $\binpRV_{\sim \ug}$, respectively.
Furthermore, 
we let $\pi({\bm x}) \equiv \pi({\bm x}')$, and use prime notation to denote a replicate set of samples, e.g. ${\bm x} \sim \pi({\bm x}) \implies {\bm x}' \sim \pi({\bm x})$.

Next, consider that 
\begin{eqnarray}
\V_{\binpRV_{\ug}} \left( f\left(\binpRV \right) | \binpRV_{\sim \ug} \right) =
\E_{\binpRV_{\ug}} \left[f^2\left( \binpRV \right) | \binpRV_{\sim \ug} \right]
- \E^2_{\binpRV_{\ug}} \left[f\left( \binpRV \right) | \binpRV_{\sim \ug} \right].
\end{eqnarray}
Therefore, 
\begin{eqnarray}\label{eq:g_tot_derv}
\E_{\binpRV_{\sim \ug}} \left[ \V_{\binpRV_{\ug}} \left( f\left(\binpRV \right) | \binpRV_{\sim \ug} \right) \right] =
\E_{\binpRV_{\sim \ug}} \left[ 
\E_{\binpRV_{\ug}} \left[f^2\left( \binpRV \right) | \binpRV_{\sim \ug} \right] \right]
- 
\E_{\binpRV_{\sim \ug}} \left[ 
\E^2_{\binpRV_{\ug}} \left[f\left( \binpRV \right) | \binpRV_{\sim \ug} \right]
\right].
\end{eqnarray}
Now consider that the second term of~\eqref{eq:g_tot_derv} can be expanded as 
\begin{eqnarray}\nonumber
&&\E_{\binpRV_{\sim \ug}} \left[ 
\E_{\binpRV_{\ug}} \left[f\left( \binpRV \right) | \binpRV_{\sim \ug} \right]
\E_{\binpRV_{\ug}} \left[f\left( \binpRV \right) | \binpRV_{\sim \ug} \right]
\right] \\\nonumber
= &&
\E_{\binpRV_{\sim \ug}} \left[
\int \int f\left( {\bm x}_{\sim \ug}, {\bm x}_{\ug}\right) 
f\left( {\bm x}_{\sim \ug}, {\bm x}'_{\ug}\right) 
\pi({\bm x}_{\ug}) d{\bm x}_{\ug} \pi({\bm x}'_{\ug})  d{\bm x}'_{\ug}
\right] \\
= &&
\int \left(
\int \int f\left( {\bm x}_{\sim \ug}, {\bm x}_{\ug}\right) 
f\left( {\bm x}_{\sim \ug}, {\bm x}'_{\ug}\right) 
\pi({\bm x}_{\ug}) d{\bm x}_{\ug} \pi({\bm x}'_{\ug})  d{\bm x}'_{\ug}
\right) \pi({\bm x}_{\sim \ug})d{\bm x}_{\sim \ug}\\\label{eq:term2}
=&&
\int \int f\left( {\bm x} \right) f\left( {\bm x}_{\sim \ug}, {\bm x}'_{\ug}\right) 
\pi({\bm x}) d{\bm x} \pi({\bm x}'_{\ug})d{\bm x}'_{\ug}.
\end{eqnarray}
Now consider the first term in \eqref{eq:g_tot_derv}, can be expanded as 
\begin{eqnarray}\nonumber
\E_{\binpRV_{\sim \ug}} \left[ 
\E_{\binpRV_{\ug}} \left[f^2\left( \binpRV \right) | \binpRV_{\sim \ug} \right] \right] &=&
\E_{\binpRV_{\sim \ug}} \left[ 
\int f^2({\bm x}_{\sim \ug}, {\bm x}_{\ug})\pi({\bm x}_{\ug})d{\bm x}_{\ug}
\right]\\\nonumber
&=&
\int \int f^2({\bm x}_{\sim \ug}, {\bm x}_{\ug})\pi({\bm x}_{\ug})d{\bm x}_{\ug} \pi({\bm x}_{\sim \ug})d{\bm x}_{\sim \ug} \\\label{eq:term1}
&=&
\int f^2({\bm x})\pi({\bm x})d{\bm x}.
\end{eqnarray}
We can then add a constant into \eqref{eq:term1} as
\begin{eqnarray}\label{eq:term1_2}
\int \int f^2({\bm x})\pi({\bm x})d{\bm x} \pi({\bm x}'_{\ug})d{\bm x}'_{\ug}.
\end{eqnarray}
Finally, we add \eqref{eq:term2} and \eqref{eq:term1_2} to represent the numerator of the grouped total Sobol' index as 
\begin{eqnarray}\nonumber
\E_{\binpRV_{\sim \ug}} \left[ \V_{\binpRV_{\ug}} \left( f\left(\binpRV \right) | \binpRV_{\sim \ug} \right) \right] &=& 
\int \int \left( f^2\left({\bm x} \right) - 
f\left( {\bm x}\right) 
f\left( {\bm x}_{\sim \ug}, {\bm x}'_{\ug}\right) 
\right)
\pi({\bm x})d{\bm x} \pi({\bm x}'_{\ug})d{\bm x}'_{\ug}
\\\label{eq:total_num_est}
&\approx&
\frac{1}{2N} \sum_{i=1}^N \left( f(\x^{(i)}) - f(\x^{(i)}_{\sim \ub},\x'^{(i)}_\ub) \right)^2, \quad \x^{(i)}, \x'^{(i)} \sim \pi(\x).
\end{eqnarray}
Note that~\eqref{eq:total_num_est} results from the fact that 
\begin{eqnarray*}
    \frac{1}{2}\left( f(\x^{(i)}) - f(\x^{(i)}_{\sim \ub},\x'^{(i)}_\ub) \right)^2 = \frac{1}{2} \left( f^2(\x^{(i)}) - 2f(\x^{(i)})f(\x^{(i)}_{\sim \ub},\x'^{(i)}_\ub) + f^2(\x^{(i)}_{\sim \ub},\x'^{(i)}_\ub) \right),
\end{eqnarray*}
where $\frac{1}{2}\left(f^2(\x^{(i)}) + f^2(\x^{(i)}_{\sim \ub},\x'^{(i)}_\ub) \right)$ is an estimator for $f^2({\bm x})$, and $f(\x^{(i)})f(\x^{(i)}_{\sim \ub},\x'^{(i)}_\ub)$ an estimator for $f({\bm x})f({\bm x}_{\sim \ug}, {\bm x}'_{\bm u})$. 
Furthermore, from \eqref{eq:total_num_est} one can understand the term \textit{pick-and-freeze} as $f(\x^{(i)}_{\sim \ub},\x'^{(i)}_\ub)$  differs from $f(\x^{(i)}) = f(\x_{\sim \ug}^{(i)}, \x_{\ug}^{(i)})$ in that the group of inputs $\ug$ are from the replicate samples ${\x}'$ rather than ${\x}$; we've ``picked'' the group $\sim \ug$ to be ``frozen'', i.e., from the same samples ${\x}$.

A similar derivation produces the estimator for the numerator of the grouped first-order index given in \eqref{eq:grouped_main_effect}:
\begin{align}
    \Var_{\X_\ub}\left(\E_{\X_{\sim \ub}}\left[f(\X)|\X_\ub\right]\right) &\approx
    \frac{1}{N}\sum_{i=1}^N f(\x^{(i)})\left( f(\x'^{(i)}_{\sim \ub}, \x^{(i)}_\ub) - f(\x'^{(i)}) \right), \quad \x^{(i)}, \x'^{(i)} \sim \pi(\x).
\end{align}
See \cite{prieur_2017} for details.
Finally, we use the two replicate sample sets $\left\{\x^{(i)}, f(\x^{(i)})\right\}_{i=1}^N$ and $\left\{\x'^{(i)}, f(\x'^{(i)})\right\}_{i=1}^N$ to estimate the total variance used in the denominators of the Sobol' index estimators as
\begin{align}
    \hat{\mu}_f &\approx \frac{1}{N}\sum_{i=1}^N f(\x^{(i)}), \nonumber \\
    \hat{\mu}'_f &\approx \frac{1}{N}\sum_{i=1}^N f(\x'^{(i)}), \nonumber \\
    V(f(\X)) &\approx \frac{1}{2N}\sum_{i=1}^N \left(f(\x^{(i)}) - \hat{\mu}_f \right)^2 + \left(f(\x'^{(i)}) - \hat{\mu}'_f \right)^2.
\end{align}

%% file: Appendix_B.tex
\section{Proofs of Propositions~\ref{prop:bounds} and~\ref{prop:bounds_2}}
\label{sec:Appendix_B}

Here, we prove Propositions~\ref{prop:bounds} and~\ref{prop:bounds_2}.

\begin{cor}
    \label{cor:2}
    $\bigg| \textnormal{Var}_{\X_\vb} \left( \E_{\X_\ub} \left[ f(\binpRV) | \X_\vb \right] \right) - 
    \textnormal{Var}_{\X_\vb} \left( \E_{\Xtilde_\ubtilde} \left[ q(\Xtilde) | \X_\vb \right] \right)\bigg| \leq \errone + \errtwo$
\end{cor}

\begin{proof}
    Recall the Law of Total Variance which states 
    \begin{eqnarray*}
        \V(f(\binpRV)) = \E_{\X_\vb} \left[ \V_{\X_\ub} \left(f(\binpRV) | \X_\vb\right) \right]
        + \V_{\X_\vb} \left( \E_{\X_\ub} \left[ f(\binpRV) | \X_\vb \right] \right)
    \end{eqnarray*}
    Then, using triangle inequality, and linearity of the expectation we have that 
    \begin{eqnarray}\nonumber
        &&\bigg| \V \left( \E \left[ f(\binpRV) | \X_\vb \right]\right) -
        \V \left( \E \left[ q(\Xtilde) | \X_\vb \right]\right)\bigg| \\\nonumber
         =&& \bigg|
        \left(\V(f(\binpRV)) - \V(q(\Xtilde))\right) 
        - \left(
            \E \left[ \V \left(f(\binpRV) | \X_\vb\right) \right]
            -
            \E \left[ \V\left(q(\Xtilde) | \X_\vb\right) \right] 
        \right)
        \bigg| \\\label{eq:prop1_1}
        \leq &&
        \bigg|\V(f(\binpRV)) - \V(q(\Xtilde)) \bigg|
        +
        \bigg| \E \left[ \V \left(f(\binpRV) | \X_\vb\right) 
        -
        \V \left(q(\Xtilde) | \X_\vb\right) \right] \bigg|,
    \end{eqnarray}
    where we have removed subscripts on the expectation and variance for notational ease.
    From Assumption (2) in Proposition~\ref{prop:bounds},
    we have that \eqref{eq:prop1_1} is 
    \begin{eqnarray}\nonumber
        < && \errtwo + \bigg| \E \left[ \V \left(f(\binpRV) | \X_\vb\right) 
        -
        \V \left(q(\Xtilde) | \X_\vb\right) \right] \bigg| 
        \\\nonumber
        \leq &&
        \errtwo + \E \left[ \bigg| 
        \V \left(f(\binpRV) | \X_\vb\right) 
        -
        \V \left(q(\Xtilde) | \X_\vb\right)
        \bigg|\right],
    \end{eqnarray}
    by Jensen's inequality.
    Corollary~\ref{cor:2} follows from Assumption (1) and the fact that $E\left[|\cdot |\right]$ is the $\mathcal{L}_1$-norm.
\end{proof}

We now prove Proposition~\ref{prop:bounds} as follows:
\begin{proof}
   
    \begin{eqnarray}\nonumber
        \big| \suf_{\ug} - \sg_{\ug'}\big|
        && =
        \left| \frac{\V \left( \E \left[ f(\binpRV)| \X_\ub\right] \right)}{\V(f(\binpRV))} - \frac{\V \left( \E \left[ q(\Xtilde)| \Xtilde_\ubtilde\right] \right)}{\V(q(\Xtilde))}\right|
        \\\nonumber
        && = 
        \Bigg|
        \frac{\V(q(\Xtilde))\V\left( \E \left[ f(\binpRV)| \X_\ub\right] \right)
        -
        \V(q(\Xtilde))\V\left( \E \left[ q(\Xtilde)| \Xtilde_\ubtilde\right] \right)}
        {\V(f(\binpRV))\V(q(\Xtilde))}
         \\\nonumber
        &&+
        \frac{
        \V(q(\Xtilde))\V \left( \E \left[ q(\Xtilde)| \Xtilde_\ubtilde\right]\right)
        -
        \V(f(\binpRV))\V \left( \E \left[ q(\Xtilde)| \Xtilde_\ubtilde\right] \right)
        }{\V(f(\binpRV))\V(q(\Xtilde))}
        \Bigg| \\\nonumber
        && \leq
        \left| 
            \frac{\V(q(\Xtilde)) \left(
            \V\left( \E \left[ f(\binpRV)| \X_\ub\right] \right)
        -
        \V\left( \E \left[ q(\Xtilde)| \Xtilde_\ubtilde\right] \right)\right)}
        {\V(f(\binpRV))\V(q(\Xtilde))}
        \right|
        \\\label{eq:prop2_1}
        &&- 
        \left|
            \frac{
        \V \left( \E \left[ q(\Xtilde)| \Xtilde_\ubtilde\right]\right) 
        \left( \V(q(\Xtilde)) - \V(f(\binpRV))\right)}
        {\V(f(\binpRV))\V(q(\Xtilde))}
        \right|,
    \end{eqnarray} 
    by the triangle inequality. Then  \eqref{eq:prop2_1} is 
    \begin{eqnarray}\nonumber
       = && \frac{\left| \V\left( \E \left[ f(\binpRV)| \X_\ub\right] \right)
       -
       \V\left( \E \left[ q(\Xtilde)| \Xtilde_\ubtilde\right] \right)\right|}{\V(f(\binpRV))}
       +
       \frac{\sg_{\ug'}\left| \V(q(\Xtilde)) - \V(f(\binpRV))  \right|}{\V(f(\binpRV))}\\\label{eq:prop2_2}
       = &&
       \left| \V\left( \E \left[ f(\binpRV)| \X_\ub\right] \right)
       -
       \V\left( \E \left[ q(\Xtilde)| \Xtilde_\ubtilde\right] \right)\right|
       + \sg_{\ug'}\left| \V(q(\Xtilde)) - \V(f(\binpRV))  \right|
    \end{eqnarray}
    by Assumption (3).
    Thus, by Corollary~\ref{cor:2} and Assumption (2), \eqref{eq:prop2_2} is 
    \begin{eqnarray*}
        < && \errtwo + \errone + \sg_{\ug'}\errtwo.
    \end{eqnarray*}
    Since all Sobol' indices are bounded above by $1$ under the assumption of independence between groups, we can conclude that
    \begin{eqnarray}
        \big| \suf_{\ug} - \sg_{\ug'}\big| < 2\errtwo + \errone.
    \end{eqnarray} 
\end{proof}

Next, we prove Proposition~\ref{prop:bounds_2} as follows:
\begin{proof}
    \begin{eqnarray}\nonumber
        \left| \ts_{\ug} - \ts_{\ug'} \right| 
        = &&
        \left| 
        \frac{\E \left[ \V \left(f(\binpRV) | \X_\vb\right) \right]}{\V(f(\binpRV))} 
        -
        \frac{\E \left[ \V \left(q(\Xtilde) | \X_\vb\right) \right]}{\V(q(\Xtilde))}
        \right| \\\nonumber
        = &&
        \bigg| 
        \frac{\E \left[ \V \left(f(\binpRV) | \X_\vb\right) \right] \V(q(\Xtilde))
        -
        \E \left[ \V \left(q(\Xtilde) | \X_\vb\right) \right]\V(q(\Xtilde))}{\V(f(\binpRV))\V(q(\Xtilde))} 
        \\\nonumber
        + &&
        \frac{
        \E \left[ \V \left(q(\Xtilde) | \X_\vb\right) \right]\V(q(\Xtilde))
        -
        \E \left[ \V \left(q(\Xtilde) | \X_\vb\right) \right]\V(f(\binpRV))
        }
        {\V(f(\binpRV))\V(q(\Xtilde))}
        \bigg| 
        \\\nonumber
        \leq &&
        \left| 
        \frac{\V(q(\Xtilde)) \left(
        \E \left[ \V \left(f(\binpRV) | \X_\vb\right) \right] 
        -
        \E \left[ \V \left(q(\Xtilde) | \X_\vb\right) \right]\right)}{\V(f(\binpRV))\V(q(\Xtilde))} 
        \right|
        \\\label{eq:prop3_1}
        + &&
        \left|
        \frac{
        \E \left[ \V \left(q(\Xtilde) | \X_\vb\right) \right]
        \left(\V(q(\Xtilde)) - \V(f(\binpRV))\right)
        }
        {\V(f(\binpRV))\V(q(\Xtilde))}
        \right|,
    \end{eqnarray}
    by the triangle inequality. 
    Again, we note Assumption (3), which, along with Assumption (1) allows us to write \eqref{eq:prop3_1} as 
    \begin{eqnarray}\nonumber
        = && \left| \E \left[ \V \left(f(\binpRV) | \X_\vb\right) \right] 
        -
        \E \left[ \V \left(q(\Xtilde) | \X_\vb\right) \right] \right| + \ts_{\ug'} \left| \V(q(\Xtilde)) - \V(f(\binpRV))\right| \\\nonumber
        <&& \errone + \ts_{\ug'} \errtwo,
    \end{eqnarray} 
    Again, since the total Sobol' indices are bounded above by $1$ under the assumption of independence between groups, we conclude that 
    \begin{eqnarray}
        \left| \ts_{\ug} - \ts_{\ug'} \right| 
        < \errone + \errtwo.
    \end{eqnarray}
\end{proof}

%% file: Appendix_C.tex
\section{Data-consistent inversion}
\label{sec:Appendix_C}

Here, details are provided on the use of data-consistent inversion (DCI) to formulate the sampling distribution on the linear operator MFU parameters -- needed to evaluate robustness (see,~\Cref{sec:comp_sobol_indices}) for the contaminant transport problem.
Given an initial distribution on parameters $\pi_0(\Xtu)$, the DCI solution is defined as
\begin{align}
  \pi_{update}(\Xtu) = \pi_0(\Xtu) \frac{\pi_{target}(q(\Xtu))}{\pi_{predict}(q(\Xtu))},
  \label{eq:dataconsistent_update}
\end{align}
where $\pi_{target}$ is the target probability density of the model output using the fractional derivative MFU representation, and 
$\pi_{predict}$ is the probability density of the pushforward of the initial density using the general linear operator representation.
The solution to~\eqref{eq:dataconsistent_update}, $\pi_{update}(\Xtu)$, has the property that, for any set $A$ in the Borel $\sigma$-algebra defined on the output space, $\mathcal{B}_D$, 
\begin{align*}
  \mathbb{P}_{update}(q^{-1}(A)) \equiv \int_{q^{-1}(A)} \pi_{update}(\Xtu) \d \Xtu = \int_A \pi_{target}( q ) \d q \equiv \mathbb{P}_{target}(A), \quad \forall A\in \mathcal{B}_D.
\end{align*}
That is, the pushforward of the updated density is consistent with the target density.
However,  
in practice, it is necessary to approximate the update using, e.g., rejection sampling, as described in Algorithm 2 of~\cite{butler2018}. 
As a result,
the output statistics produced using $\pi_{update}(\Xtu)$ will be similar (but not identical) to the output statistics produced using the fractional derivative MFU representation.

Implementation of DCI requires two key elements: approximating $\pi_{target}$ and determining $\pi_0(\Xtu)$. 
To approximate $\pi_{target}$, we construct a kernel-density estimate (KDE) using $10^3$ samples of $f(\X)$ with $\X_{\vg}$ set to their mean values and $\X_\ub$ sampled according to \Cref{eq:nu_dist,eq:alpha_dist}.
We defined $\pi_0(\Xtu)$ using a multivariate normal approximation of the fractional derivative eigenvalues as follows.
Given samples of $\nu_m$ and $\alpha$, eigenvalues of the fractional derivative operator are computed according to~\eqref{eq:fractional}.
The distribution of fractional derivative eigenvalues can then be approximated by a multivariate normal distribution using the sample mean and covariance of the samples, where the multivariate normal distribution is constructed with respect to the following unraveled and transformed values of the eigenvalues: $\big[\log(-\Re[\lambda_1]), \log(-\Re[\lambda_2]), \ldots, \log(\Im[\lambda_1]), \log(\Im[\lambda_2]), \ldots]\big]$. 
Thus, we generate samples from $\pi_0(\Xtu)$, as follows:
\begin{itemize}
  \item Sample $R_1, R_2, \ldots, I_1, I_2, \ldots \sim \mathcal{N}(\hat{\mu}, \hat{\Sigma})$. 
  \item $\lambda_k = -\exp(R_k) + i \exp(I_k)$.  
\end{itemize}

Histograms of the QoI from the target distribution and before and after DCI are shown in \Cref{fig:data_consistent_update}.
Note that the histograms of $\pi_{target}$ and the pushforward of the data-consistent update, $\pi_{update,Q}$ are very similar. 
Additionally, the updated density is computed with the other sources of uncertainty in the model (the model parameters) set at their mean values; as they vary during VBSA, we expect there to be differences between the model outputs statistics for the two MFU representations due to sampling the model parameters.
The goal of this DCI procedure is to produce similar output distributions for the purpose of numerically verifying our robustness bounds, so this is  desirable.

\begin{figure}[h!]
  \includegraphics{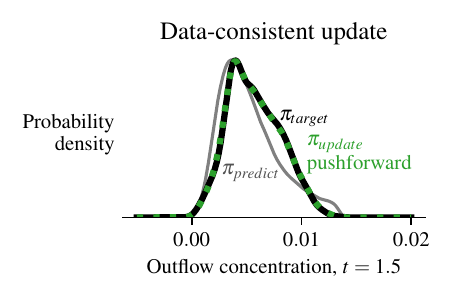}\vspace{-1em}
  \caption{Kernel density estimates for the three densities described herein. The pushforward of the data-consistent update to the QoI is pictured here, rather than the updated distribution itself.}
  \label{fig:data_consistent_update}
\end{figure}

\FloatBarrier